\documentclass[onecolumn, 11pt]{IEEEtran}
\usepackage{cite}
\usepackage{multicol}
\usepackage{amsfonts, amsmath, amssymb, amsthm}
\usepackage[utf8]{inputenc}
\usepackage{array, makecell}
\usepackage{mathtools}
\usepackage{dsfont}
\usepackage{tikz, color, xcolor, soul}
\usepackage{pgfplots}
\usepackage{bbold}
\usepackage{enumitem}
\usepackage{hyperref}

\newtheorem{theorem}{Theorem}

\newtheorem{lemma}[theorem]{Lemma}
\newtheorem{corollary}{Corollary}
\newtheorem{remark}{Remark}
\theoremstyle{definition}
\newtheorem{definition}{Definition}
\newcommand{\X}{\mathcal{X}}
\newcommand{\Y}{\mathcal{Y}}
\newcommand{\Z}{\mathbb{Z}}
\renewcommand{\P}{\mathcal{P}}
\newcommand{\Q}{\mathcal{Q}}
\newcommand{\R}{\mathbb{R}}
\newcommand{\K}{\mathcal{K}}
\newcommand{\C}{\mathcal{C}}
\newcommand{\D}{\mathcal{D}}

\renewcommand{\v}{\boldsymbol{v}}

\newcommand{\wt}{\textrm{wt}}

\renewcommand{\u}{\boldsymbol{u}}
\newcommand{\x}{\boldsymbol{x}}
\renewcommand{\L}{\mathcal{L}}
\newcommand{\y}{\boldsymbol{y}}
\newcommand{\z}{\boldsymbol{z}}
\newcommand{\0}{\boldsymbol{0}}

\renewcommand{\epsilon}{\varepsilon}
\newcommand{\E}{\mathbf{E}}
\newcommand{\s}{\mathrm{supp}}

\usepackage{xcolor}

\newcommand{\rad}{\mathrm{rad}}

\renewcommand{\mid}{\,\ifnum\currentgrouptype=16 \middle\fi|\,}

\allowdisplaybreaks

\title{Codes for the Z-channel}

\author{
   Nikita~Polyanskii, 
   and~Yihan~Zhang
   \thanks{
      Nikita~Polyanskii's research was conducted in part during October 2020 - December 2021 with the Technical University of Munich and the Skolkovo Institute of Science and Technology. His work was supported by the German Research Foundation (Deutsche Forschungsgemeinschaft, DFG) under Grant No. WA3907/1-1 and the Russian Foundation for Basic Research (RFBR) under Grant No.~\mbox{20-01-00559}. 

      Yihan Zhang is supported by funding from the European Union’s Horizon 2020 research and innovation programme under grant agreement No 682203-ERC-[Inf-Speed-Tradeoff].

      This work was presented in parts at the 2022 IEEE International Symposium on Information Theory, Espoo, Finland. 
   }%
   \thanks{
      Nikita Polyanskii is with the IOTA Foundation.
      (e-mail: \href{mailto:nikitapolyansky@gmail.com}{nikitapolyansky@gmail.com}).
   }%
   \thanks{
      Yihan Zhang is with the Institute of Science and Technology Austria.
      (e-mail: \href{mailto:zephyr.z798@gmail.com}{zephyr.z798@gmail.com}).
   }%
}

\begin{document}

\maketitle

\begin{abstract}
   This paper is a collection of results on combinatorial properties of codes for the \emph{Z-channel}. 
   A Z-channel with error fraction $\tau$ takes as input a length-$n$ binary codeword and injects in an adversarial manner up to $n\tau$ \emph{asymmetric} errors, i.e., errors that only zero out bits but do not flip $0$'s to $1$'s. 
   It is known that the largest $(L-1)$-list-decodable code for the Z-channel with error fraction $\tau$ has {exponential size (in $n$)} if $\tau$ is less than a critical value that we call the \emph{{$(L-1)$-list-decoding Plotkin point}} and has constant size if $\tau$ is larger than the threshold. 
   The $(L-1)$-list-decoding Plotkin point is known to be $ L^{-\frac{1}{L-1}} - L^{-\frac{L}{L-1}} $, which equals $1/4$ for unique-decoding with $ L-1=1 $. 
   In this paper, we derive various results for the size of the largest codes above and below the list-decoding Plotkin point. 
   In particular, we show that the largest $(L-1)$-list-decodable code $\epsilon$-above the Plotkin point, {for any given sufficiently small positive constant $ \epsilon>0 $,} has size $\Theta_L(\epsilon^{-3/2})$ for any $L-1\ge1$. 
   We also devise upper and lower bounds on the exponential size of codes below the list-decoding Plotkin point. 
\end{abstract}

\newpage
\tableofcontents
\newpage


\section{Introduction}
In coding theory, the Z-channel is used to model some asymmetric data storage and transmission systems. In this binary-input binary-output channel, the symbol $0$ is always transmitted correctly, whereas the transmitted symbol $1$ can be received as $0$. 

In this paper, we consider the combinatorial setting where the encoder transmits $n$ symbols, and the maximum number of errors inflicted by an adversary is proportional to $n$. 
{The error model under consideration is also known as the adversarial/Hamming/zero error model \cite{hamming1950error,shannon1956zero} which is standard in coding theory. 
It stands in contrast to the stochastic setting in Shannon theory \cite{shannon1948mathematical} where the error model is assumed to be average-case and a vanishing probability of decoding error is desired. }
For a given word $\x\in\{0,1\}^n$, we define the Z-ball centered at $\x$ with radius $\tau n$ {(where $ \tau\in[0,1] $ and $ n\in\Z_{\ge1} $)}\footnote{{We use $ \Z_{\ge a} $ to denote the set of integers at least $a$. }}\footnote{{Formally, $ n\tau $ needs to be an integer. 
Otherwise, in the proofs of upper (resp.\ lower) bounds on the maximal code size, one can take $\tau'$ such that $ n\tau' $ is the largest (resp.\ smallest) integer less (resp.\ bigger) than $n\tau$. 
Asymptotically in $n\to\infty$, the effect of this quantization on our bounds is negligible. 
Therefore, we at times drop ceiling/floor for notational convenience. 
Similarly, other quantities such as the Hamming weight $nw$ that will show up later in the paper can be suitably rounded up or down and the same results continue to hold.} } as a set of all possible words that can be transmitted over the Z-channel with at most $\tau n$ errors such that $\x$ is received. Given $\tau$ and $n$, the main goal for $(L-1)$-list-decoding is to construct a code $\mathcal{C}\subseteq\{0,1\}^n$ such that for any $\x\in\{0,1\}^n$, the Z-ball centered at $\x$ contains at most $L-1$ codewords from $\C$. For $L=2$, we say that  $\C$ is a uniquely-decodable code tolerating a fraction $\tau$ of  (asymmetric) errors.  For $L\ge 3$,  $\C$ is referred to as an $(L-1)$-list-decodable code for the Z-channel with list-decoding radius $\tau$. Finding fundamental limits of error-correcting codes is one of the major problems in coding theory. Uniquely-decodable codes for asymmetric errors  have been  a subject for extensive studies in the numerous papers~\cite{bassalygo1965new,borden1983low,kim1959single,varshamov1965theory,klove1981upper,klove1981error,fu2003new,bose1993asymmetric,zhang2019construction,blaum-book-asymm-code,tallini2018onsome,Al-Bassam1994asymmetric}. 
{Codes for Z-channels with feedback are studied in \cite{tallini2008feedback,tallini2009correction,deppe2020,lebedev2020}. 
Bounds and constructions of codes correcting a single asymmetric error are obtained in \cite{Constantin1979asymmetric,Bose2000systematic}. }
Up to our best knowledge, there are only two papers~\cite{lebedev2020,zhang2020generalized} in the literature that discuss properties of list-decodable codes for the Z-channel. 

First, we recall some results concerning the unique-decoding case. It is known~\cite{bassalygo1965new,borden1983low} that the rate of optimal codes tolerating a fraction $\tau$ of asymmetric errors is asymptotically equal to the rate of optimal codes correcting a fraction $\tau$ of symmetric errors~\footnote{Hereafter, errors are called symmetric if any transmitted symbol from the alphabet $\{0,1\}$ can be bit-flipped.}. Hence there exist exponential-sized uniquely-decodable codes for any fraction of errors $\tau<1/4$. The Plotkin bound~\cite{plotkin1960binary} {implies} that the size of a code capable of correcting a fraction $\tau=1/4+\epsilon$ of symmetric errors is bounded above by $1+1/(4\epsilon)$. One might ask a similar question for asymmetric errors. Specifically, can we bound the size of a code $\mathcal{C}\subseteq\{0,1\}^n$ tolerating a fraction  $\tau=1/4+\epsilon$ of asymmetric errors using a function that depends only on $\epsilon$? The paper~\cite{borden1983low} claims that such a bound exists only for $\epsilon>1/12$. We disprove this statement and provide an order-optimal uniquely-decodable code of size $\Theta(\epsilon^{-3/2})$ as $\epsilon\to 0$. From the results of~\cite{levenshtein1961application,alon2018list,plotkin1960binary}, it follows that the maximal size of a code tolerating a fraction $1/4+\epsilon$ of symmetric errors is $(4\epsilon)^{-1}(1+o(1))$ as $\epsilon\to 0$.

Much less is known about list-decodable codes for the Z-channel. By~\cite{lebedev2020,zhang2020generalized}, exponential-sized (or positive-rate) $(L-1)$-list-decodable codes with list-decoding radius $\tau$ exist only for $ \tau<\tau_L $, where $ \tau_L = w_{\max} - w_{\max}^L $ and $ w_{\max} $ equals $ L^{-\frac{1}{L-1}} $. We call $\tau_L $ the \emph{$(L-1)$-list-decoding Plotkin point}. We extend the above results from unique-decoding to list-decoding and obtain the same characterization $ \Theta_L(\epsilon^{-3/2}) $ for list-decodable codes with arbitrary list size $ L-1 $ correcting $ \tau_L+\epsilon $ fraction of asymmetric errors. 
The same question for symmetric errors was also studied before. 
In a recent work~\cite{alon2018list}, the results in~\cite{levenshtein1961application,plotkin1960binary} for unique-decoding were generalized to list-decoding with any \emph{odd} list size that is at least one and the optimal code size was shown to be $ \Theta_L(\epsilon^{-1}) $. 
For \emph{even} list size, the problem seems significantly more difficult and~\cite{alon2018list} showed that the optimal code size is $ \Theta_L(\epsilon^{-3/2}) $ for list size {two}. 

\section{Overview of our results and techniques}
\label{sec:results-techniques}

This paper is a collection of results on combinatorial properties of codes for the Z-channel with adversarial errors. 
The most technically challenging part of our results has to do with obtaining the order-optimal size of codes that correct a fraction of asymmetric errors \emph{$\epsilon$-above} the Plotkin point. 
We start with the unique-decoding case and show in Sec.~\ref{sec::high error low rate codes} that the optimal size of codes which correct $ 1/4+\epsilon $ fraction of asymmetric errors is exactly $ \Theta(\epsilon^{-3/2}) $. 
This follows from an upper bound (Theorem~\ref{th::plotkin bound for the Z channel} in Sec.~\ref{sec:unique-dec-upper-bound}) and a matching construction (Theorem~\ref{th::construction of uniquely decodable codes} in Sec.~\ref{sec:unique-dec-construction}). 
We then generalize these results to list-decoding for any list size at least one. 
We show in Sec.~\ref{sec:ub-cc}-\ref{sec:construction-general} that the same bound $ \Theta_L(\epsilon^{-3/2}) $ is also optimal for list-decodable codes which correct $ \tau_L + \epsilon $ fraction of asymmetric errors, where $ \tau_L $ is the list-decoding Plotkin point. 
See Sec.~\ref{sec:listdec-rad} and~\ref{sec:listdec-plotkin-pt} for definitions and properties of the list-decoding radius and list-decoding Plotkin point.

We briefly explain below the ideas behind our results on {codes that correct a fraction of errors beyond the Plotkin bound}. 
\begin{enumerate}
  \item (Upper bound for (approximate) constant-weight codes.) For \emph{constant-weight} codes (i.e., a code in which all codewords have the same Hamming weight; {see \cite{Graham1980lower}}), it follows from a standard double-counting argument that the size of any code $\epsilon$-above the Plotkin point (i.e., $ 1/4 $) is at most $ O_L(\epsilon^{-1}) $. 
  For unique-decoding, this is known by~\cite{plotkin1960binary};  
  for list-decoding, this follows from~\cite{lebedev2020} (see Theorem~\ref{thm:upper-bound-constant-weight} in Sec.~\ref{sec:ub-cc}). 
  Furthermore, these results can be extended to \emph{approximate constant-weight} codes (i.e., a code in which all codewords have \emph{approximately} the same Hamming weight). 
  This can be done by either repeating the double-counting argument with additional care of the deviation in weights or carefully augmenting the codewords in such a way that the augmented codewords all have the same weight. 
  See Theorem~\ref{thm:upper-bound-apx-constant-weight} and~\ref{th: converse non-constant-weight} in Sec.~\ref{sec:ub-apx-cc} for details. 

  \item (Upper bound for general codes.) For \emph{general} codes in which codewords can have any weight between $ 0 $ and $ n $, given the above results, it is tempting to partition $ \{0,1\}^n $ into $ \epsilon^{-1} $ slices each of weight between $ nw $ and $ n(w+\epsilon) $ for some $ w\in[0,1] $. 
  The subcode in each slice is therefore approximate constant-weight and has size at most $ O_L(\epsilon^{-1}) $. 
  In total, we get an upper bound $ O_L(\epsilon^{-2}) $ on the size of the whole code. 
  However, this bound turns out to be suboptimal!
  We improve it by partitioning the space in a more delicate way. 
  The width of the slice is wider for weights far from the critical {(relative)} weight $ w_{\max}=L^{-\frac{1}{L-1}}$ and is thinner for {(relative)} weights close to $ w_{\max} $.\footnote{{For the convenience of discussion, we refer ``weight'' to the Hamming weight (i.e., number of nonzero coordinates) or the relative Hamming weight (i.e., fraction of nonzero coordinates). 
    Without further clarification, the meaning shall be clear from the context.} } 
  In particular, we choose the width to be $ \epsilon^{1/2} $ on average for $ w\in[0,w_{\max} - \Omega_L(\epsilon^{1/2})]\cup[w_{\max} + \Omega_L(\epsilon^{1/2}),1] $ and we choose it to be $ \epsilon $ for $ w\in[w_{\max}-O_L(\epsilon^{1/2}),w_{\max} + O_L(\epsilon^{1/2})] $, while keeping the subcode in each slice having size $ O_L(\epsilon^{-1}) $. 
  In total, there are at most $ O_L(1/\epsilon^{1/2}) = O_L(\epsilon^{-1/2}) $ slices for small and large weights (i.e., weights far from $ w_{\max} $) and at most $ O_L(\epsilon^{1/2}/\epsilon)=O_L(\epsilon^{-1/2}) $ slices for moderate weights (i.e., weights close to $ w_{\max} $). 
  This gives an improved upper bound $ O_L(\epsilon^{-1/2}\epsilon^{-1}) = O_L(\epsilon^{-3/2}) $ on the size of the whole code. 
  The rigorous analyses for $L=2$ and $L>2$ are presented in the proofs of Theorem~\ref{th::plotkin bound for the Z channel} in Sec.~\ref{sec:unique-dec-upper-bound} and Theorem~\ref{thm:upper-bound-general} in Sec.~\ref{sec:ub-general}, respectively. 

  \item (Construction of constant-weight codes.) In Sec.~\ref{sec:construction-cc}, we analyze a code formed by rows of a matrix whose columns are all possible constant-weight words. In Theorem~\ref{th: construction constant weight list decoding}, such a code is shown to be $\epsilon$-above the list-decoding Plotkin point and have order-optimal size $ \Omega_L(\epsilon^{-1}) $. This code generalizes a construction proposed in~\cite{alon2018list} in the context of list-decodable codes for symmetric errors and has a similar flavour as weak-flip codes discussed in~\cite{lin2018weak}. 

  \item (Construction of general codes.) We note that {purely random codes that correct a fraction of errors larger than the Plotkin bound} do not have large size. However, we show how to use randomness in order to build a code $\epsilon$-above the Plotkin point of size $\Omega_L(\epsilon^{-3/2})$. The non-uniform partition used in the proof of the converse bound suggests a matching construction. 
  We reuse the constant-weight construction in Theorem~\ref{th: construction constant weight list decoding} in a consistent way with the non-uniform partition. Specifically, we first build $\Theta_L(\epsilon^{-1/2})$ constant-weight codes such that the $i$th code has size $\Theta_L(\epsilon^{-1})$ and contains codewords with relative weight  $w_{\max} - i\epsilon$. The asymmetry property of the Z-channel comes to play when we apply $\Theta_L(\epsilon^{-1/2})$ independent random permutations on the set of coordinates within each code and consider the union of all these codes. For $L=2$, we carefully analyze the unique-decoding radius of the resulting construction in the proof of Theorem~\ref{th::construction of uniquely decodable codes} in Sec.~\ref{sec:unique-dec-construction}. For $L>2$, we investigate the list-decoding radius of our construction in the proof of Theorem~\ref{th::construction of list-decodable codes}. 
\end{enumerate}

In Sec.~\ref{sec:ub-cap} and~\ref{sec:lb-cap}, we study codes for the Z-channel \emph{below} the list-decoding Plotkin point and obtain upper and lower bounds on the list-decoding capacity. 
The upper bound in Theorem~\ref{thm:eb} follows the classical idea of Elias and Bassalygo~\cite{bassalygo1965new}. Specifically, the space is multicovered by special balls and the size of all balls is carefully adjusted such that each of them contains a constant number of codewords only. The lower bound in Theorem~\ref{thm:rand-cod-lb} uses the standard random coding with expurgation technique. 
The question of obtaining the exact list-decoding capacity for arbitrary list sizes is difficult and our upper and lower bounds do not match. 
However, we manage to derive the list-decoding capacity for asymptotically large list sizes in Theorem~\ref{thm:listdec-cap-large} in Sec.~\ref{sec:listdec-cap-large}. 

Sec.~\ref{sec:cap-stoch-z} contains discussion on the capacity of Z-channels with \emph{stochastic} errors. 
This is a direct consequence of the seminal channel coding theorem by Shannon~\cite{shannon1948mathematical}. 

We end the paper with open questions in Sec.~\ref{sec:open}. 

{
\section{Relation to prior works}
\label{sec:relation}
The $(L-1)$-list-decoding Plotkin point was studied in \cite{zhang2020generalized} for a large family of adversarial channels. 
Specialized to Z-channels considered in the current paper, \cite{zhang2020generalized} provided a characterization of $ \tau_L $ which reads as follows. 
For a finite set $\Sigma$, let $ \P(\Sigma) $ denote the set of distributions on $ \Sigma $. 
For any given $ L\in\Z_{\ge2} $ and $\tau\in[0,1]$, define the \emph{confusability set} $ \K(\tau)\subset\P(\{0,1\}^L) $ as
\begin{align}
\K(\tau) &\coloneqq \left\{\sum_{y\in\{0,1\}} P_{X_1,\cdots,X_L,Y = y}\in\P(\{0,1\}^L) : \begin{array}{c}
P_{X_1,\cdots,X_L,Y} \in\P(\{0,1\}^{L+1}) \\ 
\forall i\in[L],\, P_{X_i,Y}(0,1) = 0, \, P_{X_i,Y}(1,0)\le\tau
\end{array} \right\} . \notag 
\end{align}
Here, for a distribution $ P_{A,B} $, $ \sum_b P_{A,B=b} $ denotes the marginal on $A$; for $ P_{X_1,\cdots,X_L,Y} $, $ P_{X_i,Y} $ denotes the marginal on $ (X_i,Y) $. 
The confusability set should be interpreted as the set of types (i.e., empirical distributions; see Definition \ref{def:type}) of ``confusable'' $L$-tuples $ (\x_1,\cdots,\x_L)\in(\{0,1\}^n)^L $ for which there exists a common sequence $ \y\in\{0,1\}^n $ that can be obtained by changing at most $\tau$ fraction of ones to zeros in each $ \x_i $ ($ i\in[L] $). 
It can be shown \cite{zhang2020generalized} that $ \K(\tau) $ is convex and satisfies $ \K(\tau)\subset\K(\tau') $ for any $ 0\le\tau\le\tau'\le1 $. 
Define also the set of \emph{completely positive tensors} $ \mathsf{CP}\subset(\R_{\ge0}^2)^{\otimes L} $ as
\begin{align}
\mathsf{CP} &\coloneqq \left\{\sum_{i = 1}^k p_i^{\otimes L}\in(\R_{\ge0}^2)^{\otimes L} : k\in\Z_{\ge1}, \, (p_1,\cdots,p_k)\in(\R_{\ge0}^2)^k\right\} . \notag 
\end{align}
It is known that $\mathsf{CP}$ is a convex cone (see \cite[Theorem 6.9]{qi2017tensor}). 
With the above definitions, \cite[Theorem 47, 54]{zhang2020generalized} imply that the $(L-1)$-list-decoding Plotkin point $ \tau_L $ for Z-channels is equal to 
\begin{align}
\tau_L &= \inf\{\tau\in[0,1] : \mathsf{CP}\cap\P(\{0,1\}^L)\subset\K(\tau)\} . \label{eqn:char-zhang} 
\end{align}
Though the characterization \eqref{eqn:char-zhang} is dimension-free (i.e., independent of $n$), it is formulated as an optimization problem involving checking inclusion relation between $ 2^L $-dimensional sets. 
Since evaluation of the RHS of \eqref{eqn:char-zhang} seems challenging, \eqref{eqn:char-zhang} does not directly provide an explicit expression of $ \tau_L $. 
\cite{lebedev2020} later derived the Plotkin point of constant-weight codes from the first principle following the methodology in \cite{blinovsky1986bounds} (see also Sec.\ \ref{sec:listdec-plotkin-pt}). 
The current work differs from \cite{zhang2020generalized,lebedev2020} in the following aspects. 
\begin{enumerate}
   \item Both \cite{lebedev2020,zhang2020generalized} focused on constant-weight (or more generally constant-type; see Definition \ref{def:type}) codes. 
   This is without loss of generality for studying achievable rates of codes below the Plotkin point. 
   However, above the Plotkin point, constant-weight codes exhibit different behaviours from weight-unconstrained codes. 
   We characterize the scaling of the optimal sizes of both code ensembles.
   \item 
   The converse bound in \cite[Theorem 54]{zhang2020generalized} implies that constant-weight codes correcting $ \tau_L+\epsilon $ fraction of asymmetric errors have size at most $ K(\epsilon^{-1}) $ for some function $ K(\cdot) $ that grows as least as fast as the hypergraph Ramsey number \cite{ramsey,erdos-rado-hypergraph-ramsey} which is far from being optimal. 
   Then the relation \eqref{eqn:prior-upper-bound-const-wt} was proved in \cite[Lemma 4]{lebedev2020} for constant-weight codes which, as we show in Sec.\ \ref{sec:ub-cc}, implies the optimal upper bound $ O_L(\epsilon^{-1}) $. 
   \item No construction of codes above the Plotkin point was given in \cite{lebedev2020,zhang2020generalized}. 
   Our work provides sharp lower bounds for such codes with constant or arbitrary weight by extending the construction in \cite[Theorem 1]{alon2018list}. 
   \item For codes below the Plotkin point, \cite[Theorem 3]{lebedev2020} also derived lower bounds on the achievable rates. 
   Here we follow the approach in \cite[Theorem 47]{zhang2020generalized} and derive an alternative expression of achievable rates which is formulated more compactly using the confusability set $ \K(w,\tau) $ (see \eqref{eqn:def-k} for the definition). 
   No upper bound on the achievable rates was given in \cite{lebedev2020}. 
   Here we derive an upper bound (see Theorem \ref{thm:eb}) using ideas from \cite{bassalygo1965new}. 
   We also characterize the list-decoding capacity when the list size is sufficiently large. 
\end{enumerate}
}

\section{Preliminaries}\label{sec::preliminaries}
We use the standard Bachmann--Landau notation. 
For a set $\mathcal{S}$ and an integer $0\le t\le|\mathcal{S}|$, $\binom{\mathcal{S}}{t}\coloneqq\{\mathcal{S}'\subset\mathcal{S}:|\mathcal{S}'|=t\}$.
The set $\{i,i+1,\ldots,j\}$ for some integers $i$ and $j$ with $ i\le j$ will be denoted as $[i,j]$. The set $[1,j]$ is shortly denoted as $[j]$.  By slight abuse of notation, we write $[w_1,w_2]$ to also denote the closed interval of real numbers between $w_1$  and $w_2$. A vector of length $n$ is denoted by bold lowercase letters, such as $\x$,
and the $i$th entry of the vector $\x$ is referred to as $x_i$.  The all-zero vector, whose length will be clear from the context, is written as $\0$. Define the asymmetric function $\Delta(\x,\y)$ to be the number of positions $i\in[n]$ such that $x_i=1$  and $y_i=0$. The Hamming distance between $\x,\y\in\{0,1\}^n$ is then $d_{\mathrm{H}}(\x,\y):=\Delta(\x,\y)+\Delta(\y,\x)$. The (Hamming) weight of $\x\in\{0,1\}^n$ is $\wt(\x):=d_{\mathrm{H}}(\x,\0)$; the relative weight is $\wt(\x)/n$. 
The \emph{Z-distance} between $ \x,\y\in\{0,1\}^n $ is defined as $ d_{\mathrm{Z}}(\x,\y)\coloneqq \max(\Delta(\x,\y),\Delta(\y,\x)) $. 
Note that the Hamming distance and the Z-distance are related by the following identity: $ d_{\mathrm{Z}}(\x,\y) = \frac{1}{2}(d_{\mathrm{H}}(\x,\y) + |\wt(\x) - \wt(\y)|) $. 
By this relation, we see that $ d_{\mathrm{Z}}(\cdot,\cdot) $ is indeed a distance, and it equals $d_{\mathrm{H}}(\cdot,\cdot)/2 $ if the Hamming weights of $\x$ and $\y$ are the same. 
Define the \textit{Z-ball} and the \textit{H-ball} centered at $\x\in\{0,1\}^n$ with radius $t$ as follows
\begin{align*}
B^{\text{Z}}_{t}(\x)&\coloneqq\{\y\in\{0,1\}^n:\ \Delta(\y,\x)=0,\ \Delta(\x,\y)\le t\},\\
B^{\text{H}}_{t}(\x)&\coloneqq\{\y\in\{0,1\}^n:\ d_{\mathrm{H}}(\x,\y)\le t\}.
\end{align*}
Similarly, define the \emph{Z-sphere} and the \emph{H-sphere} centered at $\x\in\{0,1\}^n$ with radius $t$ as follows
\begin{align*}
S^{\mathrm{Z}}_{t}(\x)&\coloneqq\{\y\in\{0,1\}^n:\ \Delta(\y,\x)=0,\ \Delta(\x,\y)= t\},\\
S^{\mathrm{H}}_{t}(\x)&\coloneqq\{\y\in\{0,1\}^n:\ d_{\mathrm{H}}(\x,\y)= t\}.
\end{align*}
{A (binary block) \textit{code} $\mathcal{C}\subseteq\{0,1\}^n$ is an arbitrary subset of vectors of the same length $n$.} The \textit{size} of a code $\C$ is denoted as $|\C|$.  
The \emph{rate} of $\C$ is defined as $ R(\C)\coloneqq \frac{1}{n}\log|\C| $. 
A code $\C\subseteq\{0,1\}^n$ is called \textit{$w$-constant-weight} if the weight of all codewords $\x\in\C$ is $\wt(\x)=nw$.  

\begin{remark}
With a slight abuse of terminology, we will interchangeably use weight/distance/radius to refer to \emph{relative} weight/distance/radius, that is, the former quantities normalized by $1/n$. 
Without further specification, the exact meaning should be clear from the context. 
\end{remark}

\begin{definition}[Uniquely-decodable code]
  We say that a code $\mathcal{C}\subseteq\{0,1\}^n$ \textit{corrects} $t$ asymmetric (symmetric) errors if for any $\x\in\{0,1\}^n$ the respective Z-ball (H-ball) centered at $\x$ with radius $t$ contains at most one codeword from $\C$, i.e., it holds that $|B^{\text{Z}}_{t}(\x)\cap \C| \le 1$ ($|B^{\text{H}}_{t}(\x)\cap \C|\le 1$).
\end{definition}
Note that a code $\C$ corrects $t$ asymmetric errors if and only if for any two distinct codewords $\x,\y\in\C$, it holds that $d_{\mathrm{Z}}(\x,\y)>t$; {see, e.g., \cite{weber1992necessary,tallini2008new}}. 

A code $\C\subseteq\{0,1\}^n$ is said to correct a \textit{fraction} $\tau$ of asymmetric (symmetric) errors if it corrects $t$ asymmetric (symmetric) errors for $t=\lceil \tau n \rceil$.

In the following statement, we first recall a observation from~\cite{varshamov1965theory}.
\begin{lemma}\label{lem: trivial statement}
  Let a code $\C\subseteq\{0,1\}^n$ be $w$-constant-weight.  Then $\C$ corrects  $t$ asymmetric errors if and only if it corrects $t$ symmetric errors which is in turn satisfied if and only if the minimum distance of $\C$ is larger than $2t$.
\end{lemma} 
We state two classical coding theory results, which were proved in~\cite{plotkin1960binary} and~\cite{bassalygo1965new} respectively. 
\begin{lemma}\label{lem::plotkin bound} Let  $\C\subseteq\{0,1\}^n$ be a code that corrects $t$ symmetric errors, where {$4t+3>n$}. Then its size is bounded above as follows
  $$
  |\C|\le 2\left\lfloor \frac{2t+2}{4t+3 -n}\right\rfloor.
  $$
\end{lemma}

\begin{lemma}\label{lem: bas bound} Let a code $\C\subseteq\{0,1\}^n$ be $w$-constant-weight and correct $t$ symmetric errors. {If the inequality $w^2 - (w-t)n >0$ is fulfilled, then the size of $\C$ is bounded above as follows
  $$
  |\C|\le 
  \frac{tn}{w^2 - (w-t)n}.
  $$
}
\end{lemma}

We introduce below the concept of list-decodable codes.
\begin{definition}[List-decodable code] 
  \label{def:listdec}
  Let $L$ be a positive integer at least $2$.
  A code $\mathcal{C}\subseteq\{0,1\}^n$ is said to be $(L-1)$-list-decodable with list-decoding radius $\tau$ if for any $\x\in\{0,1\}^n$ the respective Z-ball (H-ball) centered at $\x$ with radius $t:=\lceil \tau n\rceil $ contains at most $L-1$ codewords from $\C$, i.e., it holds that $|B^{\text{Z}}_{t}(\x)\cap \C| \le L-1$ ($|B^{\text{H}}_{t}(\x)\cap \C|\le L-1$).
\end{definition}

At last, we need the binary entropy function defined as $ H(x) \coloneqq -x\log x-(1-x)\log(1-x) $ for any $ x\in[0,1] $.

\section{Uniquely-decodable codes above the Plotkin point}\label{sec::high error low rate codes}
In this section, we obtain upper and lower bounds on the size of an optimal code capable of correcting a fraction $\tau=\frac{1}{4}+\epsilon$ of asymmetric errors. {For simplicity of arguments, we naturally assume in this section that if an $n$-length code corrects a fraction $\tau=\frac{1}{4}+\epsilon$ of errors, then  $\tau n = n/4+\epsilon n$ is an integer and, consequently, $\epsilon n \ge 1/4$.}
\subsection{Upper bound}
\label{sec:unique-dec-upper-bound}
In the following statement, we derive an upper bound on the size of a code capable of correcting a large fraction of asymmetric errors. The idea of the proof is to partition a code into $O(\epsilon^{-1/2})$ subcodes with approximate constant-weight. By lengthening codewords within each subcode, we obtain constant-weight codes correcting a large fraction of symmetric errors and show that their size can be bounded by $O(\epsilon^{-1})$.
\begin{theorem}\label{th::plotkin bound for the Z channel}
  Let $\C\subseteq\{0,1\}^n$ be a code correcting a fraction $\tau=\frac{1}{4}+\epsilon$ of asymmetric errors for some real number {$\epsilon\le 1/12$, where $\tau n$ is an integer}. Then the size of the code can be bounded as follows
  {
  $$
  |\C|\le\frac{1+3\sqrt{\epsilon} + 4\epsilon}{\epsilon^{3/2}}+14.
  $$}
\end{theorem}
{
\begin{remark}
A suboptimal upper bound $ O(\epsilon^{-2}) $ was presented in the first version of \cite{lebedev2020v1} for uniquely-decodable codes. The authors of the current manuscript then managed to improve this bound to the optimal order $ O(\epsilon^{-3/2})$. 
(In fact, we proved the upper bound $ O_L(\epsilon^{-3/2}) $ for $(L-1)$-list-decodable codes with any $ L\ge2 $; see Theorem \ref{thm:upper-bound-general}.)
During the finalization of the current manuscript, the first author updated the other manuscript \cite{lebedev2020} and presented the improved bound $ O(\epsilon^{-3/2})$ for $L=2$ there. 
However, the latter bound was first obtained by the authors of the current manuscript. 
\end{remark}
}

\begin{proof}[Proof of Theorem \ref{th::plotkin bound for the Z channel}]
  Without loss of generality, we assume that the number of codewords of weight at most $n/2$ is at least as the number of codewords of weight at least $n/2$. Otherwise, we can consider the code that is obtained by replacing $0$ by $1$ and $1$ by $0$ in all codewords of the original code because the Z-distance $d_{\mathrm{Z}}(\x,\y)=\max(\Delta(\x,\y),\Delta(\y,\x))$ is not changed after such swapping. For a non-negative integer $i$, define $\rho_i:=\frac{i}{2i+1}$.  Define the subcode
  \begin{equation*}\label{eq: first codes}
  \mathcal{A}_i:=\{\x\in\C:\ \lfloor \rho_{i}n \rfloor<\wt(\x)\le \lfloor \rho_{i+1}n\rfloor \}.
  \end{equation*}
  Extend all codewords of $\mathcal{A}_i$ by appending $\lfloor \rho_{i+1}n\rfloor -\lfloor \rho_{i}n \rfloor -1$ extra positions such that all codewords have weight  $\lfloor \rho_{i+1}n\rfloor $. Note that this can be done in different ways. From Lemma~\ref{lem: trivial statement} it follows that {the obtained code $\mathcal{A}_i'\subseteq\{0,1\}^{n'}$ with $n_i'\coloneqq n+\lfloor \rho_{i+1}n\rfloor -\lfloor \rho_{i}n \rfloor -1$ contains codewords of weight  $w_i'\coloneqq\lfloor \rho_{i+1}n\rfloor $ and corrects $t\coloneqq (1/4+\epsilon)n$ symmetric errors. To apply  Lemma~\ref{lem: bas bound} for code $\mathcal{A}_i'$, one needs to check that the condition of this lemma is fulfilled, specifically $w_i'^2-(w_i'-t)n_i' > 0$. Note that the function $f_i(z)\coloneqq z^2 - (z-t)n_i'$ of argument $z$ is monotonically decreasing for $z\le n_i'/2$.   Clearly, $\rho_{i+1}n\le n_i'/2$. Thus, 
  \begin{align*}
  	f_i(w_i')&=f_i\left(\lfloor\rho_{i+1}n\rfloor\right)\\
  	&\ge f_i\left(\rho_{i+1}n\right)\\
  	&=\rho_{i+1}^2n^2-(\rho_{i+1}n-(1/4+\epsilon)n)(n+ \lfloor n\rho_{i+1}\rfloor - \lfloor n\rho_{i}\rfloor -1)\\
  	&\ge n^2\left(\rho_{i+1}^2-(\rho_{i+1}-(1/4+\epsilon))(1+ \rho_{i+1} - \rho_{i})\right),
  \end{align*}
  where the last inequality holds for $\rho_{i+1}-(1/4+\epsilon) \ge 0$ or $\epsilon \le 1/12$ as $\rho_{i+1}\ge 1/3$. 
  Observe that $\rho_{i+1}^2 - (\rho_{i+1}-1/4)(1+\rho_{i+1}-\rho_i)= 0$ since $\rho_i=\frac{i}{2i+1}$. Therefore, we can proceed as follows
  $$
  f_i(w_i')\ge n^2\epsilon(1+\rho_{i+1}-\rho_{i})>0.
  $$
}

{Thus, we can apply Lemma~\ref{lem: bas bound} for code $\mathcal{A}_i'$ and estimate its size
\begin{align*}
	|\mathcal{A}_i|&=|\mathcal{A}_i'|\\
	&\le \frac{tn_i'}{f_i(w_i')}
	\\
	&\le \frac{(1/4+\epsilon)n ( n+\lfloor \rho_{i+1}n\rfloor -\lfloor \rho_{i}n \rfloor -1)}{n^2\epsilon(1+\rho_{i+1}-\rho_{i})}
	\\
	&\le 1+1/(4\epsilon).
\end{align*}
} 


  Let $i_0:= \lfloor1/(2\sqrt{\epsilon}) \rfloor $ and, thus, $\rho_{i_0}=\frac{i_0}{2i_0+1}\ge \frac{1-2\sqrt{\epsilon}}{2+2\sqrt{\epsilon}}$. For a non-negative integer $j$, consider the subcode
  \begin{equation*}\label{eq: second codes}
  \mathcal{B}_j:=\{\x\in\C:\ \lfloor\rho_{i_0}n\rfloor+j\lceil 2\epsilon n \rceil< \wt(\x)\le\lfloor\rho_{i_0}n\rfloor+(j+1)\lceil2\epsilon n\rceil\}.
  \end{equation*}
  We extend all codewords of $\mathcal{B}_j$ by appending $\lceil2\epsilon n\rceil$ extra positions such that all obtained codewords have the same weight. We get the code $\mathcal{B}_j''\subseteq\{0,1\}^{n''_j}$ with $n''\coloneqq n+\lceil2\epsilon n\rceil$ which contains codewords of weight $w''_j\coloneqq \lfloor\rho_{i_0}n\rfloor+(j+1)\lceil2\epsilon n\rceil$ 
  and corrects $t = (1/4+\epsilon)n$ symmetric errors.  {We shall apply Lemma~\ref{lem::plotkin bound} for code $\mathcal{B}_j''$. The condition of that lemma is satisfied as long as $4t+3 > n''$ or
  $$
  (1+4\epsilon)n+3 > n+  2\epsilon n + 1 \ge n+ \lceil 2\epsilon n\rceil,
$$
which is correct.  By Lemma~\ref{lem::plotkin bound}, we estimate
 \begin{align*}
	|\mathcal{B}_j|&=|\mathcal{B}_j''|\\
	&\le  2 \frac{2t+2}{4t+3 -n''}\\
	&\le \frac{n+4\epsilon n+4}{n+4\epsilon n+3-n-2\epsilon n -1} \\
	& = \frac{n}{2\epsilon n+2}+2\\
	&\le 2+1/(2\epsilon).
\end{align*}
}
  Each codeword of $\C$, having weight within the interval  $[1,n/2]$, is included in  $\mathcal{A}_i$ with some $i\in[0,i_0-1]$ or $\mathcal{B}_j$ with some $j\in[0,j_0-1]$ if $j_0:=\lceil 3/(4\sqrt{\epsilon}) + 2\rceil$. {Indeed, all codewords with weight  in interval $[1, \lfloor\rho_{i_0}n\rfloor+j_0\lceil2\epsilon n\rceil]$ are considered by this partition and the largest weight is at least $n/2$:
  \begin{align*}
  \lfloor\rho_{i_0}n\rfloor+j_0\lceil2\epsilon n\rceil &\ge \frac{1-2\sqrt{\epsilon}}{2+2\sqrt{\epsilon}} n -1 + \left( \frac{3}{4\sqrt{\epsilon}} + 2\right) 2\epsilon n \\
 &= \frac{(4\sqrt{\epsilon} - 8\epsilon)n -8\sqrt{\epsilon}-8\epsilon+12\epsilon n+12\epsilon^{3/2}n+32\epsilon^{3/2}n+32\epsilon^2 n}{(2+2\sqrt{\epsilon})4\sqrt{\epsilon}} \\
 &\ge \frac{n(4\sqrt{\epsilon}+4\epsilon)}{8\sqrt{\epsilon}+8\epsilon} + \frac{-8\sqrt{\epsilon}-8\epsilon +11\sqrt{\epsilon}+ 8\epsilon}{8\sqrt{\epsilon}+8\epsilon} > \frac{n}{2},
  \end{align*}
where in the second-last inequality we applied the fact that $\epsilon n \ge 1/4$.
}
  
  Since the number of codewords with weight from the interval $[1,n/2]$ is not less than the number of codewords of weight from the interval $[n/2,n-1]$, it holds
  {
  \begin{align*}
  |\C|&\le 2\left(\sum_{i=0}^{i_0-1} |\mathcal{A}_i| + \sum_{j=0}^{j_0-1}|\mathcal{B}_{j}|\right) +2\\
  &\le 2i_0 \left(1+\frac{1}{4\epsilon}\right)+2j_0\left(2+\frac{1}{2\epsilon}\right)+2\\
  &\le  \frac{1}{\sqrt{\epsilon}}\left(1+\frac{1}{4\epsilon}\right)+\left(\frac{3}{4\sqrt{\epsilon}}+3\right)\left(4+\frac{1}{\epsilon}\right)+2\\
  &= \frac{1/4+\epsilon+3/4+3\sqrt{\epsilon}+3\epsilon }{\epsilon^{3/2}}+14\\
  &=\frac{1+3\sqrt{\epsilon} + 4\epsilon}{\epsilon^{3/2}}+14.\qedhere
  \end{align*}
}
\end{proof}

\subsection{Construction}
\label{sec:unique-dec-construction}
In the following statement, we prove that there exists a code of size $\Omega(\epsilon^{-3/2})$ and length $\exp(\Theta(\epsilon^{-3/2}))$ capable of correcting a fraction $1/4+\epsilon$ of asymmetric errors. We use the intuition from the proof of Theorem~\ref{th::plotkin bound for the Z channel}. First, we build $O(\epsilon^{-1/2})$ codes such that the $j$th code is a uniquely-decodable code containing codewords with the relative weight $\frac{1}{2}-j\epsilon$.  We borrow an idea for such code with $j=0$ from~\cite{alon2018list}, where the authors constructed list-decodable codes for symmetric errors. By performing simple repetition,  we construct longer codes of the same size while the error-correction capability of all those codes remains unchanged. Applying a random permutation on the set of coordinates within each code, we guarantee that two codewords from different codes have a large asymmetric distance with overwhelming probability.
\begin{theorem}\label{th::construction of uniquely decodable codes} 
  There exists a code of length $\exp(O(\epsilon^{-3/2}))$, capable of correcting a fraction $\tau=\frac{1}{4}+\epsilon$ of asymmetric errors. Furthermore, its size is at least $\frac{3\sqrt{3}}{128}\epsilon^{-3/2}(1+o(1))$ as $\epsilon\to 0$.
\end{theorem}
\begin{proof}
  Consider the positive integer $m:=\lfloor 3/(32\epsilon)\rfloor$  and define the real number $c:=2^{-3/2}$. For every $j\in\{-\lfloor c\sqrt{m}\rfloor,\ldots, \lfloor c\sqrt{m}\rfloor\}$, denote $n_j:=\binom{2m}{m-j}$. Consider a binary matrix  $A_j$ of size $2m\times n_j$, whose columns are all possible binary vectors of length $2m$ and weight $m-j$. For two distinct rows  $\x$ and $\y$ of matrix  $A_j$, we compute the number of positions in which the rows are different. 
  {We have}
  $$
  \Delta(\x,\y)=\Delta(\y,\x)= \binom{2m-2}{m-j-1}.
  $$
  This means that a code, whose codewords are rows of $A_j$, corrects a fraction $\tau_j$ of asymmetric errors, where
  $$
  \tau_j:=\frac{\binom{2m-2}{m-j-1}-1}{\binom{2m}{m-j}}= \frac{(m-j)(m+j)}{2m(2m-1)} -\frac{1}{\binom{2m}{m-j}}=\frac{1}{4} + \frac{m/2-j^2}{4m^2-2m} -\frac{1}{\binom{2m}{m-j}}.
  $$
  Let us show that for small enough $\epsilon$, it holds that $\tau_j\ge\frac{1}{4}+\epsilon$. We have
  \begin{align*}
  \frac{m/2-j^2}{4m^2-2m} -\frac{1}{\binom{2m}{m-j}}&\ge \frac{m/2-c^2m}{4m^2-2m}-\frac{1}{\binom{2m}{m-\lfloor c\sqrt{m}\rfloor}}\\
  &= \frac{3}{32m-16} - \frac{1}{\binom{2m}{m-\lfloor c\sqrt{m}\rfloor}}\\
  &= \frac{3}{32 m} + \frac{3}{2m(32m-16)}- \frac{1}{\binom{2m}{m-\lfloor c\sqrt{m}\rfloor}}\\
  &\ge \epsilon.
  \end{align*}
  To show the last inequality, we used that $m= \lfloor3/(32\epsilon) \rfloor$, $c=2^{-3/2}$ and the fact that for sufficiently large  $m$ (small $\epsilon$), it holds
  $$
  \frac{3}{64m^2}>\frac{1}{\binom{2m}{m-\lfloor c\sqrt{m}\rfloor}}.
  $$
  
  For a positive integer $z$, consider a matrix $A_j^{(z)}$ of size $2m\times zn_j $ which is composed of $z$ copies of the matrix $A_j$. We write copies of $A_j$  to the right, i.e.,  $A_j^{(z)}=(A_j,A_j,\ldots,A_j)$. By $\tilde A_j^{(z)}$ denote a {random} matrix obtained from $A_j^{(z)}$ by applying a random permutation of its columns. This permutation is taken uniformly at random from the set of all permutations.  Note that a code whose codewords are rows of $\tilde A_j^{(z)}$ corrects the same fraction of errors as the original code obtained from $A_j$.  Define integers
  $$
  z_j:= \prod_{\substack{i=-\lfloor c \sqrt{m}\rfloor,\\ i\neq j}}^{\lfloor c\sqrt{m}\rfloor} \binom{2m}{m-i},\quad N:= \prod_{\substack{i=-\lfloor c \sqrt{m}\rfloor}}^{\lfloor c\sqrt{m}\rfloor} \binom{2m}{m-i},\quad M:=2m(2\lfloor c\sqrt{m}\rfloor+1).
  $$
  Consider a {random} matrix $A$ of size $M\times N$ containing $\tilde A_j^{(z_j)}$ as a submatrix for all  $j\in\{-\lfloor c\sqrt{m}\rfloor,\ldots, \lfloor c\sqrt{m}\rfloor\}$. We assume that matrices $\tilde A_j^{(z_j)}$ are written one below the other using the ascending order of parameter $j$. For $\epsilon\to 0$, we estimate the number of rows $M$ and the number of columns $N$ in $A$ as follows
  $$
  M=\frac{3\sqrt{3}}{128\epsilon \sqrt{\epsilon}}(1+o(1)),\quad N=\exp(\Theta(\epsilon^{-3/2})).
  $$
{Define the random code $\C$ of length $N$ and size $M$ to be the set of rows in $A$.} Let us show that for two distinct rows $\tilde \x$ and $\tilde \y$ in  $A$, the value $d_{\mathrm{Z}}(\tilde\x,\tilde\y)$ is large enough with overwhelming probability. The latter ensures that the code $C$ can correct the required fraction of asymmetric errors. More formally, let  $\tilde \x$ and $\tilde\y$ be rows from $\tilde A_j^{(z_j)}$ and $\tilde A_i^{(z_i)}$, where $j<i$. Clearly, $\wt(\tilde \x)=\frac{m-j}{2m}N>\frac{m-i}{2m}N=\wt(\tilde \y)$ and $d_{\mathrm{Z}}(\tilde\x,\tilde\y)=\Delta(\tilde\x,\tilde\y)$.  The probability distribution for the random variable $\Delta(\tilde \x, \tilde \y)$ is given by
  $$
  \Pr\{\Delta(\tilde \x, \tilde \y)=t\} = \begin{cases} \frac{\binom{\wt(\tilde \x)}{t}\binom{N -\wt(\tilde \x)}{\wt(\tilde \y) - \wt(\tilde \x)+t}}{\binom{N}{\wt(\tilde \y)}},\quad &\text{for }t\in\left\{\wt(\tilde \x)-\wt(\tilde \y),\ldots,\min\left(\wt(\tilde \x), N-\wt(\tilde \y)\right)\right\},\\
  0,\quad &\text{otherwise.}
  \end{cases}
  $$
  %
  We estimate the probability of the event that $\Delta(\tilde\x,\tilde\y)$ 
  is small as follows
  \begin{equation}\label{eq::trivial bound}
  \Pr\left\{\Delta(\tilde\x,\tilde\y)\le N \left(\frac{1}{4}+\epsilon\right)\right\} \le N \max\limits_{t\in[0,\lfloor N(\frac{1}{4}+\epsilon)\rfloor] } \Pr\left\{\Delta(\tilde\x,\tilde\y) = t\right\}.
  \end{equation}
  Let an integer $t$ be equal to $\alpha N$ for some real number  $\alpha\in\left[\frac{i-j}{2m},\ \min(\frac{m-j}{2m},\frac{m+i}{2m})\right]$. Note that $N$ and $m$ are functions of $\epsilon$. Define the function
  $$
  g_{i,j}(\alpha,\epsilon) :=  \frac{1}{N}\log \left(\Pr\left\{\Delta(\tilde\x,\tilde\y)=t\right\}\right).
  $$
  For arbitrary integers $u$ and $v$ so that $u> v \ge 1$, the binomial coefficient $\binom{u}{v}$ satisfies
  $$
  \sqrt{\frac{u}{8v(u-v)}}2^{uH(v/u)} \leq \binom{u}{v} \leq 2^{uH(v/u)}.
  $$
  Thus, for $\alpha\in\left(\frac{i-j}{2m},\ \min(\frac{m-j}{2m},\frac{m+i}{2m})\right)$, it holds that
  \begin{align*}
  g_{i,j}(\alpha,\epsilon) &\le \frac{m-j}{2m} H\left(\frac{2\alpha m}{m-j}\right)+\frac{m+j}{2m}H\left(\frac{j-i+2\alpha m}{m+j}\right)-H\left(\frac{m-i}{2m}\right)-\frac{\log\left(\frac{m^2}{2(m-i)(m+i)N}\right)}{2N}\\
  &\le r_{i,j}(\alpha,\epsilon)+\delta(\epsilon),
  \end{align*}
  where functions $r_{i,j}(\alpha,\epsilon)$ and {$\delta(\epsilon)$} are defined as follows
  $$
  r_{i,j}(\alpha,\epsilon):=\frac{m-j}{2m} H\left(\frac{2\alpha m}{m-j}\right)+\frac{m+j}{2m}H\left(\frac{j-i+2\alpha m}{m+j}\right)-H\left(\frac{m-i}{2m}\right),\quad \delta(\epsilon):=\frac{\log(2N)}{2N}.
  $$
  Using the relation $\frac{\partial H(x)}{\partial x}=\log(\frac{1-x}{x})$, we compute the derivative of $r_{i,j}(\alpha,\epsilon)$: 
  $$
  q_{i,j}(\alpha,\epsilon):=\frac{\partial r_{i,j}(\alpha,\epsilon)}{\partial \alpha}=\log\left(\frac{m-j-2\alpha m}{2\alpha m}\right) + \log\left(\frac{m+i-2\alpha m}{j-i+2\alpha m}\right).
  $$
  Let $\alpha_{i,j}=\alpha_{i,j}(\epsilon):=\frac{(m-j)(m+i)}{4m^2}$. Clearly, the function $q_{i,j}(\alpha,\epsilon)$ is positive for $\alpha<\alpha_{i,j}$, and the function $r_{i,j}(\alpha,\epsilon)\le 0$ for all required $\alpha$. Furthermore, $r_{i,j}(\alpha_{i,j},\epsilon)=0$. Since $m=\lfloor3/(32\epsilon)\rfloor$ and $-\lfloor c\sqrt{m}\rfloor\le j<i\le \lfloor c\sqrt{m}\rfloor$, we obtain that
  $\alpha_{i,j}(\epsilon)- \left(\frac{1}{4}+\epsilon\right) > 0$ and
  $$
  \alpha_{i,j}(\epsilon)- \left(\frac{1}{4}+\epsilon\right)=\frac{(m-j)(m+i)}{4m^2}- \left(\frac{1}{4}+\epsilon\right)=\frac{(i-j)m-ij-4\epsilon m^2}{4m^2}.
  $$
  Since the derivative of $r_{i,j}(\alpha,\epsilon)$ is positive for $\alpha\le 1/4+\epsilon$, we get that
  \begin{equation}\label{eq: maximization of g}
  \sup_{\frac{i-j}{2m} < \alpha \le \frac{1}{4}+\epsilon} g_{i,j}(\alpha,\epsilon) \le r_{i,j}(1/4+\epsilon,\epsilon)+\delta(\epsilon).
  \end{equation}
  Observe that we have the partial sum of the Taylor series with the remainder in Lagrange's form
  \begin{equation}\label{eq: taylor expansion}
  r_{i,j}(\alpha_{i,j},\epsilon)= r_{i,j}(1/4+\epsilon,\epsilon)+ (\alpha_{i,j} - 1/4-\epsilon) q_{i,j}(1/4+\epsilon,\epsilon) + (\alpha_{i,j}-1/4-\epsilon)^2\frac{\sigma_{i,j}(\theta,\epsilon)}{2},
  \end{equation}
  where $\sigma_{i,j}(\alpha,\epsilon):=\frac{\partial q_{i,j}(\alpha,\epsilon)}{\partial \alpha}$ and $\theta$ is some point within the interval $[1/4+\epsilon,\alpha_{i,j}]$. Let us show that the function $\sigma_{i,j}(\alpha,\epsilon)$ in the interval $(1/4+\epsilon,\alpha_{i,j})$  can be bounded below
  $$
  \frac{\sigma_{i,j}(\alpha,\epsilon)}{\log e}
  = \frac{j-m}{\alpha (m-j-2\alpha m)} - \frac{2m(m+j)}{(j-i+2\alpha m)(m+i-2\alpha m)}.
  $$
  For $\epsilon\to 0$, $m=\Theta(\epsilon^{-1})$. Thus, $\sigma_{i,j}(\theta,\epsilon)= -16\log e(1+o(1))$ as $\epsilon\to 0$. Indeed, it holds that
  $$
  \lim_{\epsilon\to 0}\sup_{1/4+\epsilon<\alpha< \alpha_{i,j}} \sigma_{i,j}(\alpha,\epsilon)=-16\log e, \quad \lim_{\epsilon\to 0}\inf_{1/4+\epsilon<\alpha< \alpha_{i,j}} \sigma_{i,j}(\alpha,\epsilon)=-16\log e.
  $$
  Now we estimate $q_{i,j}(1/4+\epsilon,\epsilon)$ as $\epsilon\to 0$
  \begin{align*}
  q_{i,j}(1/4+\epsilon,\epsilon) &= \log\left(1+\frac{mi-4\epsilon m^2 - jm -ji}{(m/2+2\epsilon m)(m/2+j-i+2\epsilon m)}\right)\\
  &=4\log e\frac{m(i-j)-4\epsilon m^2-ji}{m^2}(1+o(1)).
  \end{align*}
  Recall that $r_{i,j}(\alpha_{i,j},\epsilon)=0$. Combining the above bounds with~\eqref{eq: taylor expansion}, we get 
  \begin{align*}
  r_{i,j}(1/4+\epsilon,\epsilon) &=  - \left(\alpha_{i,j} - \frac{1}{4}-\epsilon\right)\left(\frac{m(i-j)-4\epsilon m^2-ji}{m^2} (4\log e-2\log e)\right)(1+o(1))
  \\
  &\le -\lambda \epsilon^2 + o(\epsilon^2)
  \end{align*}
  for some constant $\lambda>0$ and $\epsilon\to 0$.
  Using this inequality,  the bound $\delta(\epsilon)=o(\epsilon^{2})$ and the inequality~\eqref{eq: maximization of g}, we estimate the LHS of~\eqref{eq::trivial bound} as follows
  $$
  \Pr\left\{\Delta(\tilde\x,\tilde\y)\le  N\left(\frac{1}{4}+\epsilon\right) \right\} \le N2^{-\lambda\epsilon^2 N+o(\epsilon^2 N)}=o(1),
  $$
  since $\epsilon^2 N= \exp(\Omega(\epsilon^{-3/2}))$.
  The probability of the event that  $\max\left(\Delta(\tilde \x,\tilde \y),\Delta(\tilde \y,\tilde \x)\right)\le N(\frac{1}{4}+\epsilon)$  for some distinct rows $\tilde \x, \tilde \y$ in matrix  $A$ can be bounded above as follows
  $$
  \binom{M}{2} \max_{\substack{\tilde\x,\tilde\y\in A \\ \tilde\x\neq \tilde\y}} \Pr\left\{\max\left(\Delta(\tilde \x,\tilde \y),\Delta(\tilde \y,\tilde \x)\right)\le N\left(\frac{1}{4}+\epsilon\right)\right\}=o(1).
  $$
  It follows that for $\epsilon\to 0$, with overwhelming probability, the code composed of rows of matrix $A$ can correct a fraction $\tau=\frac{1}{4}+\epsilon$ of asymmetric errors. 
\end{proof}

\section{List-decodable codes above the Plotkin point}
\label{sec:listdec}

\subsection{List-decoding radius}
\label{sec:listdec-rad}
Fix $ L\in\Z_{\ge2} $. 
For an $L$-list $\L\coloneqq\{\x_1,\cdots,\x_L\}\in\binom{\{0,1\}^n}{L} $ of distinct vectors, define its \emph{Chebyshev radius} {with respect to} Z-distance as {the smallest $t$ such that all sequences in $\L$ are contained in a Z-ball of radius $ t $}:
{\begin{align}
\rad(\L) &\coloneqq \min\{t\in\{0,1,\ldots,n\} : \exists \y\in\{0,1\}^n,\, \L\subset B_t^{\mathrm{Z}}(\y)\} . 
\label{eqn:cheby-rad-def} 
\end{align}}
{In other words, $ \rad(\L) $ is the smallest number of asymmetric errors that are sufficient to change all sequences in $\L$ to a same sequence $\y$. }
Let $ \C\subseteq\{0,1\}^n $ be an arbitrary code for the Z-channel.
The \emph{$(L-1)$-list-decoding radius} $\tau_L(\C)$ of $\C$ is defined as 
\begin{align}
\tau_L(\C) &\coloneqq \min_{\L\in\binom{\C}{L}} \rad(\L). \notag 
\end{align}

It is easy to see that the minimum in the definition of $ \rad(\L) $ can be achieved by the following vector $\y_{\L} = (y_1,\ldots,y_n)\in\{0,1\}^n $:
\begin{align}
y_j &= \begin{cases}
1, & x_{1,j} = \cdots=x_{L,j} = 1 \\
0, & \mathrm{otherwise}
\end{cases} \label{eqn:def-center}
\end{align}
for each $ j\in[n] $. 
Here $ x_{i,j} $ denotes the $j$th entry of $\x_i $. 
{In words, the smallest Z-ball containing $\L$ is centered at the vector whose support is the intersection of the supports of all vectors in $\L$. }
Observe that since $ \s(\y_{\L})\subseteq\s(\x_i) $ for all $ i\in[L] $, we have $ d_{\mathrm{Z}}(\x_i,\y_{\L}) = |\s(\x_i)| - |\s(\y_{\L})| = \wt(\x_i) - \wt(\y_{\L}) $.\footnote{{For a vector $ \v $, we use $ \s(\v) $ to denote its \emph{support}, i.e., $ \s(\v)\coloneqq\{i:v_i\ne0\} $. }} 
Therefore 
\begin{align}
\rad(\L) \coloneqq& \max_{i\in[L]} \wt(\x_i) - \wt(\y_{\L}) = \max_{i\in[L]} \wt(\x_i) - \left|\bigcap_{i\in[L]}\s(\x_i)\right|. \notag 
\end{align}
Furthermore, if all $ \x_i $'s have the same weight $nw$, then 
\begin{align}
\rad(\L) &= nw - \left|\bigcap_{i\in[L]}\s(\x_i)\right|. \label{eqn:def-rad-cw}
\end{align}
For a constant-weight code $\C$ of weight $w$, we have
\begin{align}
\tau_L(\C) = nw - \max_{\L\in\binom{\C}{L}} \left|\bigcap_{\x\in\L}\s(\x)\right|. \notag 
\end{align}

It is not hard to see that $\C\subseteq\{0,1\}^n$ is $(L-1)$-list-decodable with list-decoding radius $\tau$ if and only if $ \tau_L(\C)>n\tau $. 
The \emph{$(L-1)$-list-decoding capacity} of a Z-channel with input weight $w$ and fraction of asymmetric errors $ \tau $ is defined as 
\begin{align}
C_{L-1}(w,\tau) &\coloneqq \limsup_{n\to\infty} \max_{\substack{\C_n\subset B^{\mathrm{Z}}_{nw}(\0) \\ \tau_L(\C_n) > n\tau}} R(\C_n). \notag 
\end{align}

The \emph{$(L-1)$-list-decoding radius} of a Z-channel with input weight $w$ and rate $R$ is defined as 
\begin{align}
\tau_L(w,R) &\coloneqq \limsup_{n\to\infty} \max_{\substack{\C_n\subset B^{\mathrm{Z}}_{nw}(\0) \\ R(\C_n)\ge R}} \tau_L(\C_n). \notag 
\end{align}

\subsection{List-decoding Plotkin point}
\label{sec:listdec-plotkin-pt}

Let $ \tau_L(w)\coloneqq w-w^L $ and $ \tau_L \coloneqq \max\limits_{0\le w\le1}\tau_L(w) $. With slight abuse of notation, we write $w_{\max}$ to denote the argument $w$ that attains the maximum of $\tau_L(w)$. 
It is not hard to check that $ \tau_L(w) $ is concave in $w$ and the maximizer $w_{\max}$ equals $ L^{-\frac{1}{L-1}} $. 
Therefore $ \tau_L = L^{-\frac{1}{L-1}} - L^{-\frac{L}{L-1}} = (1-L^{-1})L^{-\frac{1}{L-1}} = (L-1)L^{-\frac{L}{L-1}} $. 

Note that $ w_{\max} $ is increasing in $L$.
The minimum value of $ w_{\max} $ over all $ L\in\Z_{\ge2} $ is $1/2$ when $ L=2 $ and $ w_{\max}\xrightarrow{L\to\infty}1 $. 

Note that $ \tau_L $ is increasing {in} $L$ and attains its minimum value $1/4$ at $L=2$. 
Furthermore, $ \tau_L\xrightarrow{L\to\infty}1 $.

\subsection{Upper bound for constant-weight codes}
\label{sec:ub-cc}

Fix $ w\in(0,1) $. 
Consider an arbitrary code $\C$ of constant-weight $w$ for the Z-channel with noise level $\tau = \tau_L(w)+\epsilon$. 
Let $ M\coloneqq |\C| $. 
It was shown in~\cite{lebedev2020} via a double-counting argument that
\begin{align}
\frac{M^L}{M(M-1)\cdots(M-L+1)} &\ge \frac{\tau}{\tau_L(w)}. \label{eqn:prior-upper-bound-const-wt} 
\end{align}
Rearranging terms, we get
\begin{align}
&& (M(M-1)\cdots(M-L+1))(\tau_L(w)+\epsilon) &\le M^L\tau_L(w) \notag \\
\implies && (M(M-1)\cdots(M-L+1))\epsilon &\le (M^L - M(M-1)\cdots(M-L+1))\tau_L(w) \notag \\
\implies && \frac{M(M-1)\cdots(M-L+1)}{M^L - M(M-1)\cdots(M-L+1)} &\le \frac{\tau_L(w)}{\epsilon}. \notag 
\end{align}
Applying Taylor expansion to the LHS of the above inequality {(as a function of $M$) at $ M\to\infty $}, we get
\begin{align}
&& \frac{M}{\binom{L}{2}} - O_L(1) &\le \frac{\tau_L(w)}{\epsilon} \notag \\
\implies&& M &\le \frac{C_{L,w}}{\epsilon} + O_L(1), \notag 
\end{align}
where $ C_{L,w}\coloneqq \tau_L(w)\binom{L}{2} = (w-w^L)\binom{L}{2} $. 

We have therefore proved the following theorem. 
\begin{theorem}
  \label{thm:upper-bound-constant-weight}
  Any $w$-constant-weight $(L-1)$-list-decodable code $ \C $ for the Z-channel with list-decoding radius $ w-w^L+\epsilon $ has size at most $ \frac{C_{L,w}}{\epsilon} + O_L(1) $ where $ C_{L,w} \coloneqq (w-w^L)\binom{L}{2} $. 
\end{theorem}
\begin{remark}
From Theorem~\ref{thm:upper-bound-constant-weight} it follows that the  rate $R(\C)$ of a $w$-constant-weight code $\C\subseteq\{0,1\}^n$ for the Z-channel with $(L-1)$-list-decoding radius $ w-w^L+\epsilon $ is asymptotically zero as $n\to\infty$. This observation motivates us to call the discussed codes list-decodable zero-rate codes.
\end{remark}

\subsection{Construction of constant-weight codes}
\label{sec:construction-cc}
In this section, we construct codes of constant-weight $w$ whose list-decoding radius is $\epsilon $-\emph{above} the list-decoding Plotkin point $ \tau_L(w) $ for Z-channels. 
The code we construct has size $ \Theta(1/\epsilon) $, which is optimal by the upper bound in the preceding section. 
Our construction is inspired by~\cite{alon2018list}. 

We prove the following theorem. 
\begin{theorem} \label{th: construction constant weight list decoding}
  There exists a $w$-constant-weight $(L-1)$-list-decodable code $\C$ for the Z-channel with list-decoding radius $ w-w^L {+} \epsilon $ of size at least $ \frac{c_{L,w}}{\epsilon} + o_{L,w}(1) $ where $ c_{L,w} = (1-w)w^{L-1}\binom{L}{2} $. 
\end{theorem}

\begin{proof}
  Fix the relative weight $0<w<1$ to be rational without loss of generality. 
  Let $ m $ be a sufficiently large integer. 
  Since $m$ is sufficiently large, we may assume that $ m/w $ is an integer. 
  Let $ M \coloneqq m/w $ and $ n \coloneqq \binom{m/w}{m} $. 
  Our construction of $\C$ can be viewed as an $M\times n$ matrix, each row of which is a codeword denoted by $ \x_i $ ($ 1\le i\le M $). 
  The code $\C$ consists of all possible length-$(m/w)$ vectors of Hamming weight $m$ as its \emph{columns}.\footnote{Note that there are $ \binom{m/w}{m} = n $ such columns in total.} 
  
  To examine $(L-1)$-list-decodability of $\C$, we compute the $(L-1)$-list-decoding radius $ \tau_L(\C) $ of $\C$. 
  By symmetry, the Chebyshev radius $ \rad(\L) $ of any $L$-list $ \L\in\binom{\C}{L} $ of codewords does not depend on the choice of $\L$. 
  Therefore, 
  $\tau_L(\C) = \E_\L\{\rad(\L)\}$
  where the expectation is over $\L$ uniformly chosen from $ \binom{\C}{L} $. 
  We now compute the latter quantity. 
  \begin{align}
  \E_\L\{\rad(\L)\} &= nw - \E_\L\left\{\left|\bigcap_{\x\in\L}\s(\x)\right|\right\} \notag \\
  &= nw - \E_\L\left\{ \sum_{i\in[n]} \mathds1\{\forall\x\in\L,\;x_i = 1\} \right\} \notag \\
  &= nw - \sum_{i\in[n]} \Pr_\L\{\forall\x\in\L,\;x_i = 1\} \notag \\
  &= nw - n\frac{\binom{m}{L}}{\binom{m/w}{L}}. \notag 
  \end{align}
  The last equality follows since the probability in the summation can be viewed as the probability that one gets $L$ ones when sampling without replacement $L$ bits from a length-$(m/w)$ vector of weight $m$. 
  
  Taking the Taylor expansion of the above expression at $ m\to\infty $, we get 
  \begin{align}
  \tau_L(\C) = \E_\L\{\rad(\L)\} = n\left( w - w^L + (1-w)w^L\binom{L}{2}m^{-1} + O_{L,w}(m^{-2}) \right). \notag
  \end{align}
  Recall that we want $\C$ to be $ \epsilon $-above the Plotkin point, i.e., $ \tau_L(\C) \ge n(\tau_L(w) + \epsilon) = n(w - w^L + \epsilon) $. 
  This is satisfied as long as 
  \begin{align}
  && \epsilon &\le (1-w)w^L\binom{L}{2}m^{-1} + O_{L,w}(m^{-2})  \notag \\
  \impliedby && \epsilon &\le (1-w)w^{L-1}\binom{L}{2}M^{-1} + O_{L,w}(M^{-2}) \notag \\
  \impliedby && M &\le \frac{c_{L,w}}{\epsilon} + o_{L,w}(1), \notag 
  \end{align}
  where $ c_{L,w} \coloneqq (1-w)w^{L-1}\binom{L}{2} $. 
  This finishes the proof. 
\end{proof}

\subsection{Upper bound for approximate constant-weight codes}
\label{sec:ub-apx-cc}
The proof in Sec.~\ref{sec:ub-cc} can be modified to work for \emph{approximate} constant-weight codes. 
\begin{theorem}
  \label{thm:upper-bound-apx-constant-weight}
  Let $\C$ be an arbitrary $(L-1)$-list-decodable code for the Z-channel with list-decoding radius $ \tau = w-w^L + \epsilon $. 
  Suppose that every codeword in $\C$ has weight in $ [nw(1-\delta),nw(1+\delta)] $ for some $ \delta\in[0,\tau/(2w)) $. 
  Then $ M\coloneqq|\C| $ satisfies the following relation:
  \begin{align}
  \frac{M^L}{M(M-1)\cdots(M-L+1)} &\ge \frac{\tau - 2w\delta}{w(1+\delta) - (w(1-\delta))^L}. \notag 
  \end{align}
\end{theorem}

\begin{remark}
  If $ \delta = 0 $, the above theorem recovers Theorem~\ref{thm:upper-bound-constant-weight}. 
\end{remark}

\begin{proof}
  To derive an upper bound on $ M $, 
  we follow a double-counting argument commonly used in coding theory. 
  We bound from both sides the following quantity:
  \begin{align}
  \sum_{\L \in[M]^L} \sum_{j\in\L} d_{\mathrm{Z}}(\x_j,\y_{\L}), \label{eqn:dc-obj}
  \end{align}
  where $ \y_\L $ is defined in Eqn.~\eqref{eqn:def-center}. 
  
  We first give a lower bound on Eqn.~\eqref{eqn:dc-obj}. 
  We drop all terms in Eqn.~\eqref{eqn:dc-obj} with an $\L$ whose elements are not all distinct.
  For any $ \L $ whose elements are all distinct, by list-decodability, 
  \begin{align}
  \max_{j\in\L} d_{\mathrm{Z}}(\x_j,\y_\L) = \max_{j\in\L} \wt(\x_j) - \wt(\y_\L) > n\tau. \notag 
  \end{align} 
  Since $\C$ is approximate constant-weight, we have $ |\wt(\x_i) - \wt(\x_j) |\le2nw\delta $ for any $ i\ne j $. 
  Therefore, Eqn.~\eqref{eqn:dc-obj} is at least $ M(M-1)\cdots(M-L+1)(n\tau + (L-1)(n\tau - 2nw\delta)) $. 
  
  We then give an upper bound on Eqn~\eqref{eqn:dc-obj}. 
  \begin{align}
  \sum_{(i_1,\cdots,i_L)\in[M]^L} \sum_{j\in[L]} d_{\mathrm{Z}}(\x_{i_j},\y_\L) 
  &= \sum_{(i_1,\cdots,i_L)\in[M]^L} \sum_{j\in[L]} (\wt(\x_{i_j}) - \wt(\y_\L)) \notag \\
  &= \sum_{(i_1,\cdots,i_L)\in[M]^L} \sum_{j\in[L]} \left(\wt(\x_{i_j}) - \sum_{k\in[n]}\mathds1\left\{x_{i_1,k} = 1\right\}\cdots\mathds1\left\{x_{i_L,k} = 1\right\}\right) \notag \\
  &\le M^LLnw(1+\delta) - \sum_{j\in[L]} \sum_{k\in[n]} \prod_{\ell\in[L]}\left(\sum_{i_\ell\in[M]}\mathds1\left\{x_{i_\ell,k} = 1\right\}\right) \notag \\
  &= M^LLnw(1+\delta) - L\sum_{k\in[n]} S_k^L, \notag 
  \end{align}
  where $ S_k\coloneqq \sum\limits_{i\in[M]}\mathds1\left\{x_{i,k} = 1\right\} $ denotes the weight of the $k$th \emph{column} of $ \C\in\{0,1\}^{M\times n} $. 
  By the norm comparison inequality $ \|\x\|_p\le n^{1/p-1/q}\|\x\|_q $ for any $ \x\in\R^n $ and $ 0<p<q $, we have
  \begin{align}
  \sum_{k\in[n]}S_k^L &\ge n^{1-L} \left(\sum_{k\in[n]}S_k\right)^L \ge n^{1-L}(Mnw(1-\delta))^L = n(Mw(1-\delta))^L. \notag 
  \end{align}
  Therefore, Eqn.~\eqref{eqn:dc-obj} is upper bounded by $ M^LLnw(1+\delta) - nL(Mw(1-\delta))^L $. 
  
  Finally, putting the lower and upper bounds together, we get
  \begin{align}
  && M(M-1)\cdots(M-L+1)(n\tau + (L-1)(n\tau - 2nw\delta))
  &\le M^LLnw(1+\delta) - nL(Mw(1-\delta))^L \notag \\
  \implies && M(M-1)\cdots(M-L+1) (L\tau - L\cdot2w\delta) &\le M^LLw(1+\delta) - LM^L(w(1-\delta))^L \notag \\
  \implies && M(M-1)\cdots(M-L+1) (\tau - 2w\delta) &\le M^L(w(1+\delta) - (w(1-\delta))^L) \notag \\
  \implies && \frac{M^L}{M(M-1)\cdots(M-L+1)} &\ge \frac{\tau - 2w\delta}{w(1+\delta) - (w(1-\delta))^L}. \qedhere
  \end{align}
\end{proof}

In fact, one can apply a similar trick as that used in the proof of Theorem~\ref{th::plotkin bound for the Z channel} and obtain a cheaper version of Theorem~\ref{thm:upper-bound-apx-constant-weight}. 
\begin{theorem}\label{th: converse non-constant-weight}
  Let $ \C $ be an $(L-1)$-list-decodable code for the Z-channel with list-decoding radius $ \tau>\tau_L $. 
  Suppose that every codeword in $\C$ has weight between $ nw_1 $ and $ nw_2 $ for some $ 0\le w_1\le w_2\le1 $. 
  Then 
  \begin{align}
  |\C| &\le \frac{L-1}{1-\left(\frac{w_{2}/(1+w_{2}-w_{1}) - w_{2}^L/(1+w_{2}-w_{1})^L}{\tau/(1+w_{2}-w_{1})}\right)^{\frac{1}{L-1}}}. \notag 
  \end{align}
\end{theorem}

\begin{proof}
  The proof follows by augmenting the code $\C$ to reduce it from approximate constant-weight to exact constant-weight. 
  
  Let $ \C $ be as described in the theorem statement. 
  One can append $ (w_2-w_1)n $ coordinates to each codeword in $\C$ in such a way that all codewords of length $ (1+w_2-w_1)n $ have weight exactly $ w_2n $. 
  Note that the \emph{relative} weight of each codeword is now $ \frac{w_2}{1+w_2-w_1} $. 
  Moreover, the augmented code is $ (L-1) $-list-decodable with \emph{relative} list-decoding radius $ \frac{\tau}{1+w_2-w_1} $. 
  
  We now recall the upper bound for $w$-constant-weight code given by Eqn.~\eqref{eqn:prior-upper-bound-const-wt}:
  \begin{align}
  && \frac{M^L}{M(M-1)\cdots(M-L+1)} &\ge \frac{\tau}{w-w^L} \notag \\
  \implies&& (M-L+1)^{L-1}\tau &\le M^{L-1}(w-w^L) \notag \\
  \implies&& M &\le \frac{L-1}{1-\left(\frac{w-w^L}{\tau}\right)^{\frac{1}{L-1}}}. \notag 
  \end{align}
  Normalizing the parameters by $ \frac{1}{1+w_2-w_1} $, we get the desired bound. 
\end{proof}
In the following statement, we estimate the size of a code $\epsilon$-above the list-decoding Plotkin point if all codewords are of weight $\epsilon$-close to each other.
\begin{corollary}\label{cor:: close weights}
  Let $ \C $ be an $(L-1)$-list-decodable code for the Z-channel with list-decoding radius $ \tau=\tau_L+\epsilon$. Suppose that the relative weight of all codewords is in the range $[w_1,w_2]$, where $w_2-w_1\le  \phi_L\epsilon $ with $\phi_L$ being a positive real number so that $\phi_L\le \frac{1-\tau_L}{2\tau_L}$.
  Then $|\C|\le \frac{(L-1)^2}{\epsilon}$ for small enough $\epsilon$.
\end{corollary}
\begin{proof}
  By Theorem~\ref{th: converse non-constant-weight} we get that 
  $$
  |\C|\le \frac{L-1}{1-\left(\frac{w_{2}/(1+w_{2}-w_{1}) - w_{2}^L/(1+w_{2}-w_{1})^L}{\tau/(1+w_{2}-w_{1})}\right)^{\frac{1}{L-1}}}.
  $$
  After simple manipulation, it remains to show that
  $$
  w_2 - \frac{w_2^L}{(1+\epsilon \phi_L)^{L-1}}\le \left(1-\frac{\epsilon}{L-1}\right)^{L-1} (\tau_L+\epsilon).
  $$
  The LHS of the above inequality achieves its maximum at $w_2=\frac{1+\epsilon \phi_L}{L^{1/(L-1)}}$. Thus, after simplification, we get a stronger sufficient condition
  $$
  (1+\epsilon \phi_L) \tau_L \le (1-\epsilon) (\tau_L+\epsilon)
  $$
  or 
  $$
  \phi_L\le \frac{(1-\tau_L)\epsilon - \epsilon^2}{\epsilon\tau_L}.
  $$
  The latter holds for sufficiently small $\epsilon$, and the corollary follows.
\end{proof}

\begin{remark}
Readers who are familiar with the literature may have noticed that unlike results for symmetric errors \cite{blinovsky1986bounds,alon2018list}, to prove Theorem~\ref{thm:upper-bound-apx-constant-weight}, we did not use Ramsey-theoretic machinery to first extract a subcode which is (approximately) ``equidistant''. 
The benefit of the Ramsey reduction is that the Chebyshev radius of the subcode turns out to be (approximately) the same as another stronger notion of radius called ``average radius''. 
The average radius is analytically much easier to deal with since unlike the Chebyshev radius, it does not involve a minimax expression. 
However, the drawback is that the size of the subcode is much smaller than the original one. 
The optimal size of {codes correcting a fraction of symmetric errors above the Plotkin bound} remains open. 

The reason why in Theorem~\ref{thm:upper-bound-apx-constant-weight} we did not need Ramsey reduction while still managed to obtain the optimal bound on the code sizes is as follows. 
For asymmetric errors, the Chebyshev radius of (approximate) constant-weight codes admits an explicit expression since the Chebyshev center (i.e., the minimizer $\y$ in Eqn.~\eqref{eqn:cheby-rad-def}) can be easily identified with Eqn.~\eqref{eqn:def-center}. 
Note, however, that in the symmetric case there is no simple formula for the Chebyshev center. 
\end{remark}

\subsection{Upper bound for general codes}
\label{sec:ub-general}
For $L\ge 2$ and $\epsilon>0$, denote the maximal size of a list-decodable code for the Z-channel with $(L-1)$-list-decoding radius $\tau_{L}+\epsilon$ by $M_L(\epsilon)$. From Corollary~\ref{cor:: close weights} one can immediately see that $M_L(\epsilon)=O_L(\epsilon^{-2})$. However, it is possible to improve this bound. In principle, we will follow the same approach as in the proof of Theorem~\ref{th::plotkin bound for the Z channel}.  For simplicity of arguments, we assume that some numbers that are introduced in the discussion in this and next subsections are integers since this will not affect the main conclusion.
\begin{theorem}
  \label{thm:upper-bound-general}
  For $L\ge 2$, it holds that $M_{L}(\epsilon)=O_L(\epsilon^{-3/2})$ as $\epsilon\to 0$.
\end{theorem}

\begin{remark}
In the above Theorem~\ref{thm:upper-bound-general} and Theorem~\ref{th::construction of list-decodable codes} in Sec.~\ref{sec:construction-general}
below, 
we make no effort to optimize the constants implicit in $ O_L(\epsilon^{-3/2}) $ and $ \Omega_L(\epsilon^{-3/2}) $. 
Even if we did, they are likely to be suboptimal. 
An interesting open question is to characterize (or provide bounds on) the leading constant which depends on $L$. 
\end{remark}

\begin{proof}
  Let $\C$ be an $(L-1)$-list-decodable code of length $n$ with $(L-1)$-list-decoding radius $\tau(\C)=\tau_L+\epsilon$ for some $\epsilon>0$. Recall that $w_{\max}$ is the argument achieving the maximum of the function $w-w^L$ and $\tau_L=w_{\max}-w_{\max}^L$. Consider all codewords whose Hamming weight is between $w_{1}n$ and $w_{2}n$ with $w_{1}<w_{2}$. 
  
  \textbf{Small weight $w_{1}=w_{\max}-\Omega_L(\sqrt{\epsilon})$:} In the following, we analyze the case when $w_{1}=w_{\max}-\Omega_L(\sqrt{\epsilon})$, where $w_{\max}$ is the maximizer of $w-w^L$, i.e. $w_{\max}$ satisfies $1-Lw_{\max}^{L-1}=0$. In our analysis, we assume that $w_{1}$ is fixed and $\Delta:=w_{2}-w_1>0$ is to be specified. By Theorem~\ref{th: converse non-constant-weight},  the number of codewords of these weights is at most
  $$
  \frac{L-1}{1-\left(\frac{w_{2}/(1+w_{2}-w_{1}) - w_{2}^L/(1+w_{2}-w_{1})^L}{\tau(\C)/(1+w_{2}-w_{1})}\right)^{\frac{1}{L-1}}}.
  $$
  This quantity is at most $\frac{L-1}{\epsilon}$ if $\Delta=\Delta(w_{1})$ satisfies
  \begin{equation}\label{eq: key inequality}
  w_2-w_2(w_2/(1+\Delta))^{L-1} \le \tau(\C)(1-\epsilon)^{L-1}.
  \end{equation}
  Recall that $\tau(\C) = \tau_L+\epsilon = w_{\max}-w_{\max}^L+\epsilon$. For any $\gamma>0$, there exists a sufficiently small $\epsilon>0$ such that the RHS of inequality~\eqref{eq: key inequality} is lower bounded as follows
  \begin{align}\label{eq: RHS}
  \tau(\C)(1-\epsilon)^{L-1} &\ge (w_{\max}-w_{\max}^L+\epsilon)\left(1-(L-1)\epsilon\right)\\
  &\ge w_{\max}- w_{\max}^L + \epsilon\left(1-\gamma-(L-1)( w_{\max}- w_{\max}^L)\right).\nonumber
  \end{align}
  Define $C_L$ to be $1-\gamma-(L-1)( w_{\max}- w_{\max}^L)$. Clearly, for  $2\le L\le 3$ and sufficiently small $\gamma$, $C_L>0$, whereas for $L>3$, $C_L<0$. Note that $(w_1+\Delta)(1+\Delta)^{-1}\ge w_1$.  Now we elaborate  the LHS of~\eqref{eq: key inequality}:
  \begin{align}\label{eq: LHS}
  w_2-w_2(w_2/(1+\Delta))^{L-1}\le&  w_2-w_1((w_1+\Delta)(1+\Delta)^{-1})^{L-1}\\\nonumber
  \le & w_2 -w_1^L.
  \end{align}
  Given~\eqref{eq: key inequality}-\eqref{eq: LHS}, we conclude that it is sufficient to take $w_2$ such that it satisfies
  \begin{equation}\label{eq: choise of Delta}
  w_2\le 
  w_{\max} - w_{\max}^L+w_1^L+C_L\epsilon.
  \end{equation}
  Consider the iterative process
  $$
  x_i:=w_{\max}-w_{\max}^{L}+x_{i-1}^L +C_L\epsilon
  $$
  with starting point $x_{1}=0$. We shall prove by induction that $x_i\ge w_{\max}-\frac{D_L}{i}$ for all $i\in[1,\epsilon^{-1/2}]$ and some absolute constant $D_L\ge 0$. The base case $i=1$ holds true if $D_L\ge w_{\max}$. Now assume that the inductive hypothesis is true for some $i>1$. Recall the Bernoulli inequality $(1-y)^s\ge 1-sy$ for any real $y\le 1$ and $s\ge 1$. Then  we obtain  
  \begin{align*}
  x_{i+1}&\ge w_{\max}-w_{\max}^L+w_{\max}^L\left(1-\frac{D_L}{i w_{\max}}\right)^L + C_L \epsilon\\
  &\ge w_{\max}-w_{\max}^L+w_{\max}^L \left(1 - \frac{2 D_L}{i w_{\max}} +\frac{D_L^2}{i^2 w_{\max}^2}  \right)^{L/2}  + C_L \epsilon\\
  &\ge w_{\max}-w_{\max}^L+w_{\max}^L \left(1- \frac{L D_L}{i w_{\max}} + \frac{L D_L^2}{2i^2w_{\max}^2}\right) + C_L\epsilon \\
  &\ge w_{\max} -\frac{D_L}{i+1} - \frac{D_L}{i(i+1)} + \frac{D_L^2}{2 w_{\max}i^2} + C_L\epsilon.
  \end{align*}
  Note that for $D_L \ge \max(4w_{\max},|C_L|)$ and $i\le \epsilon^{-1/2}$, the sum of the last three terms is non-negative and the inductive hypothesis follows. 
  
  Define the subcode $\C'\subseteq \C$ that includes all codewords of $\C$ with the Hamming weight within the range $[0,n(w_{\max}-D_L\epsilon^{1/2})]$. The above arguments imply that the size of $\C'$ is $O_L(\epsilon^{-3/2})$.

  \textbf{Large weight $w_{1}=w_{\max}+\Omega_L(\sqrt{\epsilon})$:} 
  Consider the iterative process 
  $$
  x_i:=w_{\max}-w_{\max}^{L}+x_{i+1}^L +C_L\epsilon
  $$
  with starting point $x_{1}=w_{\max}+E_L\epsilon^{1/2}$, where an absolute constant $E_L$ is to be specified. We shall prove by induction that $x_i\ge w_{\max}+\frac{E_L}{\epsilon^{-1/2}+1-i}$ for all $i\in[1,\epsilon^{-1/2}]$. The base case $i=1$ holds true. Assume the hypothesis is true for some $i>1$. Then we obtain
  \begin{align*}
  x_{i+1} &\ge w_{\max} - w_{\max}^L +  \left(w_{\max}+\frac{E_L}{\epsilon^{-1/2}+1-i}\right)^L + C_L\epsilon
  \\
  &= w_{\max} - w_{\max}^L + w_{\max}^L \left(1 + \frac{2E_L}{w_{\max}(\epsilon^{-1/2}+1-i)} + \frac{E_L^2 }{w_{\max}^2(\epsilon^{-1/2}+1-i)^2}\right)^{L/2} + C_L\epsilon
  \\
  &\ge w_{\max} + \frac{E_L}{\epsilon^{-1/2}+1-i} + \frac{E_L^2}{2w_{\max}(\epsilon^{-1/2}+1-i)^2} + C_L\epsilon
  \\
  &= w_{\max} + \frac{E_L}{\epsilon^{-1/2}-i} + \frac{E_L}{(\epsilon^{-1/2}+1-i)(\epsilon^{-1/2}-i)}+ \frac{E_L^2}{2w_{\max}(\epsilon^{-1/2}+1-i)^2} + C_L\epsilon.
  \end{align*}
  Observe that the sum of the last three terms is non-negative for $E_L\ge\max\left\{8w_{\max}, |C_L|/2\right\}$.
  
  Define the subcode $\C''\subseteq \C$ that includes all codewords of $\C$ with the Hamming weight within the range $[n(w_{\max}+D_L\epsilon^{1/2}), n]$. The above arguments imply that the size of $\C''$ is $O_L(\epsilon^{-3/2})$.
  
  \textbf{Moderate weight $|w_1-w_{\max}|=O_L(\sqrt{\epsilon})$:} By Corollary~\ref{cor:: close weights} we can partition the set of codewords with weight $[n(w_{\max}-D_L\epsilon^{1/2}),n(w_{\max}+E_L\epsilon^{1/2})]$ into $O_L(\epsilon^{-1/2})$ subcodes such that each of them has size $O_L(\epsilon^{-1})$.
  
  Summing up the above discussion, we conclude that the size of $\C$ can be bounded as $O_L(\epsilon^{-3/2})$ as $\epsilon\to 0$.
\end{proof}

\subsection{Construction of general codes}
\label{sec:construction-general}
In this section we construct $(L-1)$-list-decodable codes of size $\Omega_L(\epsilon^{-3/2})$ and length $\exp(\Theta(\epsilon^{-3/2}))$ whose list-decoding radius is $\tau_L+\epsilon$ as $\epsilon\to 0$. 
First we prove an auxiliary technical lemma.
\begin{lemma}\label{lem: key technical lemma for random permutation}
Let $\x_1,\ldots,\x_L$ be binary vectors of length $N$ such that the Hamming weight of these vectors is $w_1N,\ldots,w_LN$ with $0\le  w_1\le \ldots \le w_L\le  1$. Define $\tilde \x_1,\ldots,\tilde \x_L$ to be random vectors obtained from $\x_1,\ldots,\x_L$ by applying independent random permutations over the set of coordinates
(each of the $N!$ permutations is equally likely to appear). Let $\tilde W N$ denote the number of coordinates $i\in[N]$ such that $\tilde x_{1,i}=\tilde x_{2,i}=\ldots = \tilde x_{L,i}=1$, i.e., random variable $\tilde W$ is defined as the fraction of coordinates where all vectors $\tilde \x_1,\ldots,\tilde \x_L$ are ones. Then for any $\gamma>0$ it holds that 
$$
\Pr\left\{\tilde W \ge \gamma + \prod_{i=1}^{L}w_i \right\} \le (L+1)\exp(-N\gamma^2 2^{-2L+1}).
$$
\end{lemma}
\begin{proof}
Fix a real number $\delta$ such that $0<\delta < \gamma/(2^{L}-1)$. Let $\overline{\x}_i$ be a random binary vector such that each coordinate of $\overline{\x}_i$ is an independent random variable which has Bernoulli distribution with parameter {$\overline{w}_i:=\min(1,w_i+\delta)$}. We note that the random vectors  $\overline{\x}_i$ and $\tilde \x_i$ can be equivalently defined (in terms of distributions) using the following three steps: 
\begin{enumerate}
	\item  sample an independent binomial random variable $\xi_i$ with parameters $\mathsf{Bin}(N,\overline{w}_i)$ and set $\eta_i$ to be constant $ w_i N$;
	\item  define $\y_i$ and $\z_i$ to be the binary vectors whose first $\xi_i$ and, respectively, $\eta_i$ coordinates are ones and the remaining coordinates are zeros;
	\item apply an independent random permutation $\pi_i$, defined over the set of coordinates $[N]$, to $\y_i$ and $\z_i$ {to obtain $\overline{\x}_i$ and $\tilde \x_i$}.
\end{enumerate}
Let $\overline{W}N$ denote the number of coordinates $i\in[N]$ such that $\overline{x}_{1,i}=\ldots=\overline{x}_{L,i}=1$, i.e., $\overline{W}$ is the fraction of coordinates where all vectors $\overline{\x}_1,\ldots,\overline{\x}_L$ {are ones}. Clearly, $\overline{W}N$ has a Binomial distribution with parameters $N$ and $\prod_{i=1}^L \overline{w}_i$. Let $A_i$ denote the event that {$\xi_i \ge  w_i N$}. Clearly, $\tilde W$ is stochastically dominated by the random variable $\overline{W}$ conditioned on events $A_1,\ldots ,A_L$. {If $\overline{w}_i=1$, then $A_i$ happens with probability $1$. Thus,} by Hoeffding's inequality, we obtain
$$
\Pr\{A_i\} \ge 1 - \exp(-2\delta^2 N).
$$
Hence, 
\begin{equation}\label{eq: intersection of A's}
\Pr\{A_1 \cap \ldots \cap A_L\} \ge 1 - L\exp(-2\delta^2 N).
\end{equation}
Then 
\begin{align}
&\Pr\left\{\tilde W \ge \gamma + \prod_{i=1}^{L}w_i \right\} \nonumber \\
\le & \Pr\left\{\overline{W} \ge \gamma + \prod_{i=1}^{L}w_i \mid A_1,\ldots,A_L \right\} \nonumber \\
\le & \frac{\Pr\left\{\overline{W} \ge \gamma + \prod_{i=1}^{L}w_i\right\}}{\Pr\left\{A_1 \cap \ldots \cap A_L\right\}} \nonumber \\
\le & 1 - \Pr\left\{A_1 \cap \ldots \cap A_L\right\} + \Pr\left\{\mathsf{Bin}(N,w') \ge N\left(\gamma + \prod_{i=1}^{L}w_i\right) \right\} \label{eq: main inequality for w tilde},
\end{align}
where $w':=\prod_{i=1}^{L}  \overline{w}_i$. One can easily prove by induction on $L$ that {$w'=\prod_{i = 1}^L \min(1, w_i + \delta) \le \prod_{i=1}^{L}w_i + (2^{L}-1)\delta$. 
Indeed, to verify 
\begin{align}
\prod_{i = 1}^L (w_i + \delta) &\le \prod_{i = 1}^L w_i + (2^L - 1)\delta \notag 
\end{align} for $ 0\le w_1\le\cdots\le w_L\le1 $ and $ \delta\in[0,1] $, first observe that equality holds when $L=1$. Then assuming the inequality for $L-1$, we have \begin{align}
\prod_{i = 1}^L (w_i + \delta) &\le \left(\prod_{i = 1}^{L-1}w_i + (2^{L-1} - 1)\delta\right) (w_L + \delta) \notag \\
&= \prod_{i = 1}^L w_i + \left[(2^{L-1} - 1)w_L + \prod_{i = 1}^{L-1}w_i + (2^{L-1} - 1)\delta\right] \delta \notag \\
&\le \prod_{i = 1}^L w_i + \left[(2^{L-1} - 1) + 1 + (2^{L-1} - 1)\right] \delta , \notag 
\end{align}
as desired. }
By Hoeffding's inequality, we get 
$$
\Pr\left\{\mathsf{Bin}(N,w') \ge N\left(\gamma + \prod_{i=1}^{L}w_i\right) \right\} \le \exp\left(-2(\gamma-(2^{L}-1)\delta)^2 N\right).
$$ 
By combining the latter inequality with inequalities~\eqref{eq: intersection of A's}-\eqref{eq: main inequality for w tilde}, we obtain
$$
\Pr\left\{\tilde W \ge \gamma + \prod_{i=1}^{L}w_i \right\} \le L\exp(-2\delta^2 N)+\exp\left(-2(\gamma-(2^{L}-1)\delta)^2 N\right).
$$ 
After choosing $\delta$ that satisfies the equality $\delta = \gamma -(2^{L}-1)\delta$, i.e., $\delta = \gamma/2^L$, we derive the required statement.
\end{proof}
Now we are ready to present the main statement concerning the existence of list-decodable codes. In principle, we follow the same arguments as used in the case of uniquely-decodable codes. The suggested construction has order-optimal size as $\epsilon\to 0$ by the upper bound for general codes.
\begin{theorem}\label{th::construction of list-decodable codes} 
  There exists an $(L-1)$-list-decodable code code of length $\exp(O_L(\epsilon^{-3/2}))$ whose list-decoding radius is $\tau_L+\epsilon$. Furthermore, its size is  $\Omega_L(\epsilon^{-3/2})$ as $\epsilon\to 0$.
\end{theorem}

\begin{proof}
  Recall that $w_{\max}$ is the argument attaining the maximum of the function $\tau_L(w) = w - w^L$, i.e., $w_{\max}$ satisfies $w_{\max^{L-1}}=\frac{1}{L}$.  Consider the positive integer $m:=1/\epsilon$. Define $J$ to be the set of consecutive integers between $-\sqrt{(w_{\max}-w_{\max}^2)m/2}$ and $\sqrt{(w_{\max}-w_{\max}^2)m/2}$. For every $j\in J$, denote $n_j:=\binom{m}{w_{\max} m-j}$. Consider a binary matrix  $A_j$ of size $m\times n_j$, whose columns are all possible binary vectors of length $m$ and weight $wm-j$. By the proof of Theorem~\ref{th: construction constant weight list decoding}, we get that the code formed by rows of matrix $A_j$ is a $w_j$-constant-weight $(L-1)$-list-decodable code with the list-decoding radius $\tau_{j,L}$, where
  $$
  \tau_{j,L} = w_j - w_j^L+(1-w_j)w_j^{L-1} \binom{L}{2} m^{-1} + O_L(m^{-2}),\quad w_j= \frac{m w_{\max} - j}{m}.
  $$
  Clearly, we have
 \begin{align}
  &\min\{\tau_{j,L}: j\in J\} \nonumber
  \\
  = &w_{\max} - w_{\max}^L  + \min_{j\in J}\left(\frac{jLw_{\max}^{L-1}}{m} - \frac{j}{m}- \frac{j^2\binom{L}{2}w_{\max}^{L-2}}{m^2}\right) +(1-w_{\max})w_{\max}^{L-1} \binom{L}{2}m^{-1}+ o_L(m^{-1}) \nonumber
  \\
  =&w_{\max} - w_{\max}^L + (1 - w_{\max})\frac{L-1}{4m} + o_L(m^{-1}) \nonumber
  \\
  = &\tau_L + \Omega_L(\epsilon). \label{eq: list-decoding radius}
  \end{align}
  For a positive integer $z$, consider a matrix $A_j^{(z)}$ of size $2m\times zn_j $ that is composed of $z$ copies of the matrix $A_j$. We write copies of $A_j$  to the right, i.e.,  $A_j^{(z)}=(A_j,A_j,\ldots,A_j)$. By $\tilde A_j^{(z)}$ denote a matrix obtained from $A_j^{(z)}$ by a random permutation of its columns. This permutation is taken uniformly at random from the set of all permutations.  Note that the list-decoding radius of the code formed by $\tilde A_j^{(z)}$ is the same as for the code given by $A_j$.  Define integers
  $$
  z_j:= \prod_{\substack{i\in J \setminus \{j\}}} \binom{m}{w_{\max}m-j},\quad N:= \prod_{\substack{i\in J}} \binom{m}{w_{\max}m-j},\quad M:=m |J|.
  $$
  Consider a random matrix $\tilde A$ of size $M\times N$ containing $\tilde A_j^{(z_j)}$ as a submatrix for all  $j\in J$. We assume that the matrices $\tilde A_j^{(z_j)}$ are written one below the other using the ascending order of parameter $j$. Let $\C$ denote a code formed by rows of $\tilde A$. For $\epsilon\to 0$, we estimate the number of rows $M$ and the number of columns $N$ in $\tilde A$ as follows
  $$
  M=\Theta_L(\epsilon^{-3/2}),\quad N=\exp(\Theta_L(\epsilon^{-3/2})).
  $$
  We claim that with high probability, $\C$ is an $(L-1)$-list-decodable code with the list-decoding radius $\tau_L + \Omega_L(\epsilon)$.  In the remainder of the proof, we prove this claim.

  Let us take $L$ distinct codewords $\L=\{\tilde \x_1,\ldots,\tilde \x_L\}$ from $\C$ such that $L_1$  codewords $\L_1=\{\tilde\x_1,\ldots,\tilde\x_{L_1}\}$ are rows from $\tilde A^{(z_{j_1})}_{j_1}$, $L_2$  codewords $\L_2=\{\tilde\x_{L_1+1},\ldots,\tilde\x_{L_1+L_2}\}$ are rows from $\tilde A^{(z_{j_2})}_{j_2}$, $\dots$, $L_k$ codewords $\L_k=\{\tilde\x_{L_1+\ldots+L_{k-1}+1},\ldots,\tilde\x_{L_1+\ldots+L_k}\}$ are rows from $\tilde A^{(z_{j_k})}_{j_k}$ with $j_1<j_2<\ldots < j_k$. We shall prove that the Chebyshev radius for $\L$ is at least $\tau_L+\Omega_L(\epsilon)$. If $k=1$, then the statement follows from~\eqref{eq: list-decoding radius}. Hereafter, assume that $k\ge 2$. For $i\in[k]$, the Hamming weight of codewords in $\L_i$ is $Nw_{j_i}$, and the number of positions where all codewords in $\L_i$ are ones is 
  $Nw'_i$ with $w'_i:=w_{j_i} - \tau_{j_i,L_i}$. Let $\tilde W N$ denote the number of coordinates $i\in[N]$ such that $\tilde x_{1,i}=\ldots=\tilde x_{L,i}=1$.   
   By applying Lemma~\ref{lem: key technical lemma for random permutation}, we get that  for any $\gamma>0$
   $$
   \Pr\left\{\tilde W \ge \gamma + \prod_{i=1}^{k}w_i' \right\} \le (k+1)\exp(-N\gamma^2 2^{-2k+1}).
   $$
Since $j_1$ is the smallest integer in the set $\{j_1,\ldots,j_k\}$, $\tilde \x_1$ has the lowest Hamming weight $Nw_{j_1}$ in the set $\L$. Therefore, the Chebyshev radius for $\L$ satisfies  $\rad(\L)=Nw_{j_1}-N\tilde W$ and hence  
 $$
 \Pr\left\{\rad(\L)\le N\left(w_{j_1} - \gamma - \prod_{i=1}^{k}w'_i\right)\right\}\le (k+1)\exp(-N\gamma^2 2^{-2k+1}).
 $$
The number of $L$-size subcodes of $\C$ is $\binom{M}{L}=O_L(\epsilon^{-3L/2})$. By the union bound, the $(L-1)$-list-decoding radius $\tau_L(\C)$ is at least 
$$
N\min_{\substack{1\le k\le L\\ j_1,\ldots,j_k\in J\\ j_1< \ldots <j_k\\ L_1,\ldots,L_k\ge 1\\L_1+\ldots +L_k=L}}\left(w_{j_1} - \gamma - \prod_{i=1}^{k}w'_i\right)
$$
with probability at least
$$
1- \binom{M}{L} (k+1)\exp(-N\gamma^2 2^{-2k+1}) = 1- O_L(\epsilon^{-3L/2}\exp(-N\gamma^2 2^{-2L+1})),
$$
where $N=\exp(\Theta_L(\epsilon^{-3/2}))$.
Therefore, it remains to show that 
$$
w_{j_1} - \gamma - \prod_{i=1}^{k}w'_i = \tau_L+\Omega_L(\epsilon)
$$
for all admissible $j_1,\ldots,j_k$, $L_1,\ldots,L_k$, small enough $\gamma$ and large enough $m$ (and consequently large $N=N(m)$). Recall that $w_{j_i}=w_{\max}-j_i/m$, $w_i'=w_{j_i}-\tau_{j_i,L_i}=w_{j_i}^{L_i}-(1-w_{j_i})w_{j_i}^{L_i-1}\binom{L_i}{2}m^{-1}+O_L(m^{-2})$ and $m=1/\epsilon$. Thus, it holds that
\begin{align*}
w_i'&=(w_{\max}-j_i\epsilon)^{L_i} - (1-w_{\max}+j_i\epsilon)(w_{\max}-j_i\epsilon)^{L_i-1}\binom{L_i}{2}\epsilon + o_L(\epsilon)\\
&=w_{\max}^{L_i}-L_ij_iw_{\max} \epsilon -(1-w_{\max})w_{\max}^{L_i-1}\binom{L_i}{2}\epsilon + o_L(\epsilon) \\
&\le w_{\max}^{L_i}-j_iL_i\epsilon+o_L(\epsilon).
\end{align*}
Recall that $\tau_L=w_{\max}-w_{\max}^L$ and $L_1+\ldots+L_k=L$. Then we obtain
\begin{align*}
    w_{j_1}-\prod_{i=1}^{k}w'_i&\ge w_{\max}-j_1\epsilon - \prod_{i=1}^{k}\left(w_{\max}^{L_i}-j_iL_i\epsilon+o_L(\epsilon)\right)\\
    &=\tau_L + \epsilon\left (-j_1+\sum_{i=1}^{k}\frac{w_{\max}^L}{w_{\max}^{L_i}}j_iL_i\right) + o_L(\epsilon)\\
    &\overset{(a)}{\ge}\tau_L+\epsilon\left(-j_1+\frac{1}{L}\sum_{i=1}^{k}j_iL_i\right)+o_L(\epsilon)\\
    &\overset{(b)}{=} \tau_L+\Omega_L(\epsilon),
\end{align*}
where $(a)$ follows from the fact $L_i\ge 1$ and  $w_{\max}^{L-1}=\frac{1}{L}$, $(b)$ is implied by the fact that $L_1+\ldots+L_k=L$, $j_i$'s are integers and $j_1<j_2<\ldots<j_k$ with $k\ge 2$.
\end{proof}

\section{List-decodable codes below the Plotkin point}
\label{sec:ld-below}
In the following two subsections, we will bound $(L-1)$-list-decoding capacity $ C_{L-1}(w,\tau) $ for $ \tau<w-w^L $. 
We note that for such a purpose, it suffices to consider constant-weight codes. 
This is because for any general code $\C$ of size $ M $, one can find a constant-weight subcode $\C'\subseteq\C$ of size at least $ M/(n+1) $. 
The rate of $\C$ and $\C'$ is asymptotically (in $n$) the same. 

Since the analyses involve applications of the method of types, we need to first introduce the notion of \emph{types} of a tuple of vectors. 
Let $ \P(\Sigma) $ denote the set of distributions supported on a finite set $ \Sigma $. 

\begin{definition}[Type]
\label{def:type}
  Let $ \Sigma_1,\cdots,\Sigma_k $ be finite alphabets. 
  Let $ (\x_1,\cdots,\x_k)\in\Sigma_1^n\times\cdots\times\Sigma_k^n $. 
  The \emph{joint type} $ T_{\x_1,\cdots,\x_k}\in\R^{|\Sigma_1|\times\cdots\times|\Sigma_k|} $ of $ \x_1,\cdots,\x_k $ is defined as their empirical distribution ({also known as} histogram), i.e., 
  for any $ (x_1,\cdots,x_k)\in\Sigma_1\times\cdots\times\Sigma_k $, 
  \begin{align}
  T_{\x_1,\cdots,\x_k}(x_1,\cdots,x_k) &\coloneqq
  \frac{1}{n}\left|\left\{i\in[n]:\x_{1,i} = x_1,\cdots,\x_{k,i} = x_k\right\}\right|. \notag 
  \end{align}
\end{definition}

Sanov's theorem determines the large deviation exponent of the type of an i.i.d.\ vector. 
\begin{theorem}[Sanov~\cite{sanov1958probability}]
  \label{thm:sanov}
  Let $ \Q\subset\P(\Sigma) $ be a subset of distributions on a finite set $ \Sigma $ such that it is equal to the closure of its interior. 
  Let $ \x\in\Sigma^n $ be distributed according to $ P^{\otimes n} $ for some $ P\in\P(\Sigma) $. 
  Then
  \begin{align}
  \lim_{n\to\infty}\frac{1}{n}\log\Pr\{T_{\x}\in\Q\} &= - \inf_{Q\in\Q} D(Q\|P), \notag 
  \end{align}
  where $ D(Q\|P) $ denotes the Kullback--Leibler divergence between $Q$ and $P$ defined as
  \begin{align}
      D(Q\|P) &\coloneqq \sum_{x\in\Sigma} Q(x)\log\frac{Q(x)}{P(x)}. \notag
  \end{align}
\end{theorem}

\subsection{Upper bound on $(L-1)$-list-decoding capacity}
\label{sec:ub-cap}

In this section, we derive an upper bound on the size of any list-decodable code for the Z-channel. 
At a high level, the proof follows the idea of Elias and Bassalygo~\cite{bassalygo1965new}. 
The idea is to \emph{cover} the space where the code is living using Hamming balls\footnote{In fact, it is more convenient for us to consider sets that are slightly more structured than Hamming balls which will be specified in the proceeding proof.}.
The radius of the balls is carefully chosen so that there can only be a \emph{constant} number of codewords in the each ball satisfying the list-decodability condition. 
Then the total number of codewords is bounded by the covering number times a constant. 
We flesh out the detail below. 

We first present a covering lemma that will be useful in the proof of the upper bound. 
\begin{lemma}
\label{lem:covering}
Let $ w,v\in(0,1) $ and $ \max(0,w+v-1)\le a\le\min(w,v) $. 
Define 
\begin{align}
I(w,v,a) &\coloneqq
(1-w-v+a)\log\frac{1-w-v+a}{(1-w)(1-v)}
+ (v-a)\log\frac{v-a}{(1-w)v}
+(w-a)\log\frac{w-a}{w(1-v)}
+a\log\frac{a}{wv}. \notag 
\end{align}
to be the mutual information of the joint distribution $ P_{U,X} \coloneqq \begin{bmatrix}
1-w-v+a & v-a \\
w-a & a
\end{bmatrix} $.\footnote{Note that the marginal $ P_X $ of $ P_{U,X} $ is $ \mathsf{Ber}(w) $ and the marginal $ P_U $ of $ P_{U,X} $ is $ \mathsf{Ber}(v) $.} 
Then for any $ \epsilon>0 $, there exists a covering $ \D\subset S_{nv}^{\mathrm{H}}(\0) $ of $ S^{\mathrm{H}}_{nw}(\0) $ satisfying: 
for any $ \x\in S^{\mathrm{H}}_{nw}(\0) $, 
there is a $ \u\in\D $ such that $ T_{\u,\x} = P_{U,X} $. 
Furthermore, the size of $\D$ is at most $ 2^{n(I(w,v,a) + \epsilon)} $. 
\end{lemma}

\begin{proof}
We will show that with high probability a random subset $\D$ of $ S^{\mathrm{H}}_{nv}(\0) $ is a covering of $ S^{\mathrm{H}}_{nw}(\0) $. 
Indeed, we sample $ M\coloneqq 2^{n(I(w,v,a) + \epsilon)} $ vectors uniformly at random from all weight-$nv$ vectors and call such a set $\D$. 
The probability that some sequence $ \x $ in $ S^{\mathrm{H}}_{nw}(\0) $ is not covered by any vector in $ \D $ can be bounded as follows:
\begin{align}
& \Pr\{\exists \x\in S^{\mathrm{H}}_{nw}(\0),\;\forall \u\in\D,\;T_{\u,\x}\ne P_{U,X}\} \notag \\
&\le \sum_{\x\in S^{\mathrm{H}}_{nw}(\0)} \prod_{\u\in\D} (1-\Pr\{T_{\u,\x}= P_{U,X}\}) \notag \\
&=  \sum_{\x\in S^{\mathrm{H}}_{nw}(\0)} \prod_{\u\in\D} \left(1-\frac{\binom{n(1-w)}{n(v-a)}\binom{nw}{na}}{\binom{n}{nv}}\right) \notag \\
&= \binom{n}{nw} \left(1-2^{n(1-w)H\left(\frac{v-a}{1-w}\right)+nwH\left(\frac{a}{w}\right)-nH(v)+o(n)}\right)^M \notag \\
&\le 2^{nH(w)} \left(1-2^{nI(w,v,a)+o(n)}\right)^{2^{nI(w,v,a)}2^{n\epsilon}} \notag \\
&\xrightarrow{n\to\infty} 2^{nH(w)} e^{-2^{n\epsilon}} \label{eqn:limit} \\
&= e^{-2^{n\epsilon}+O_L(n)}. \notag 
\end{align}
In Eqn.~\eqref{eqn:limit}, we used the fact $ \lim\limits_{n\to\infty}(1-1/n)^n = 1/e $. 
Therefore, with probability at least $ 1-e^{-2^{\Omega(n)}} $, $\D$ is a covering of $ S^{\mathrm{H}}_{nw}(\0) $ {with respect to} $ P_{U,X} $. 
This finishes the proof of the lemma. 
\end{proof}

We are now ready to prove the upper bound which reads as follows. 
\begin{theorem}
\label{thm:eb}
Fix $ 0<w<1 $. 
Let $ \C $ be a $w$-constant-weight list-decodable code for the Z-channel of length $n$ with list-decoding radius $ \tau $. 
Then for $n\to\infty$
\begin{align}
R(\C) &\le \min_{\substack{0\le v\le 1,\;\max(0,w+v-1)\le a\le\min(w,v) \\ (1-v)\left[\frac{w-a}{1-v} - \left(\frac{w-a}{1-v}\right)^L\right] + v\left[\frac{a}{v} - \left(\frac{a}{v}\right)^L\right] \le \tau}} I(w,v,a)+o(1), \notag 
\end{align}
where $ I(w,v,a) $ was defined in Lemma~\ref{lem:covering}. 
\end{theorem}

\begin{proof}
By Lemma~\ref{lem:covering}, for any list-decodable code $ \C $ with constant-weight $ nw $, one can find a weight-$nv$ covering $ \D $ {with respect to} a joint distribution $ P_{U,X} $ as specified in Lemma~\ref{lem:covering}, where the parameters $ v $ and $ a $ are to be optimized later and $\D$ satisfies the properties given by the lemma. 

For each $ \u\in S^{\mathrm{H}}_{nv}(\0) $, define the following \emph{jointly typical set} $ A(\u,P_{U,X}) $ {with respect to} the joint distribution $ P_{U,X} $:
\begin{align}
A(\u,P_{U,X}) &\coloneqq \left\{\x\in S^{\mathrm{H}}_{nw}(\0):T_{\u,\x} = P_{U,X} \right\}. \label{eqn:A-def} 
\end{align}
By the covering property of $\D$, we have 
\begin{align}
&& S^{\mathrm{H}}_{nw}(\0) &= \bigcup_{\u\in\D} A(\u,P_{U,X}) \notag \\
\implies&& S^{\mathrm{H}}_{nw}(\0)\cap\C &= \left(\bigcup_{\u\in\D} A(\u,P_{U,X})\right)\cap\C \notag \\
\implies&& \C &= \bigcup_{\u\in\D} (A(\u,P_{U,X})\cap\C) \notag \\
\implies&& |\C| &\le \sum_{\u\in\D} |A(\u,P_{U,X})\cap\C|. \notag 
\end{align}
For each $ \u\in\D $, define $ \C_{\u} \coloneqq A(\u,P_{U,X})\cap\C $. 
By Markov's inequality, there must be a $ \u^*\in\D $ such that $ |\C_{\u^*}|\ge|\C|/|\D| $. 
We further define the punctured subcodes of $ \C_{\u^*} $ {with respect to} $ \u^* $. 
For $ u\in\{0,1\} $, let 
$\C_{\u^*,u} \coloneqq \{(x_i)_{i\in[n],u^*_i = u}:\x\in\C\}$ be the subcode obtained by restricting codewords in $\C$ to the coordinates $i$'s where $ u^*_i = u $. 
Note that $ \C_{\u^*, 1}\in\{0,1\}^{nv} $ and $ \C_{\u^*, 0}\in\{0,1\}^{n(1-v)} $. 

The punctured subcodes $ \C_{\u^*,0} $ and $ \C_{\u^*,1} $ enjoy the following property. 
For $ u\in\{0,1\} $, all codewords in $ \C_{\u^*,u} $ have the same type $ P_{X|U = u} $. 
That is, every codeword in $ \C_{\u^*,0} $ has weight $ w-a $ and every codeword in $ \C_{\u^*,1} $ has weight $ a $. 
Clearly, 
$ |\C_{\u^*}|\le|\C_{\u^*,0}|\cdot|\C_{\u^*,1}| $. 

Suppose that there are $ n\tau_0 $ errors in the locations $i$'s such that $ u^*_i = 0 $ and there are $ n\tau_1 $ errors in the locations $ j $'s such that $ u^*_j = 1 $. 
The parameters $ \tau_0 $ and $ \tau_1 $ satisfy $ \tau_0 + \tau_1 \le \tau $. 
As long as the following two conditions are satisfied, 
\begin{align}
\frac{\tau_0}{1-v} &> \frac{w-a}{1-v} - \left(\frac{w-a}{1-v}\right)^L, \quad
\frac{\tau_1}{v} > \frac{a}{v} - \left(\frac{a}{v}\right)^L, \notag 
\end{align}
both $ |\C_{\u^*,0}| $ and $ |\C_{\u^*,1}| $ are at most a constant (independent of $n$) by Theorem~\ref{thm:upper-bound-constant-weight}. 
Specifically, if 
\begin{align}
\frac{\tau_0}{1-v} = \frac{w-a}{1-v} - \left(\frac{w-a}{1-v}\right)^L+\epsilon, \quad
\frac{\tau_1}{v} = \frac{a}{v} - \left(\frac{a}{v}\right)^L+\epsilon, \notag 
\end{align}
then 
\begin{align}
|\C_{\u^*,0}| &\le \frac{C_{0}}{\epsilon} + O_L(1), \quad 
|\C_{\u^*,1}| \le \frac{C_{1}}{\epsilon} + O_L(1), \notag
\end{align}
where 
\begin{align}
C_0 &\coloneqq \left[\frac{w-a}{1-v} - \left(\frac{w-a}{1-v}\right)^L\right]\binom{L}{2}, \quad 
C_1 \coloneqq \left[\frac{a}{v} - \left(\frac{a}{v}\right)^L\right]\binom{L}{2}. \notag 
\end{align}
Therefore 
\begin{align}
|\C| &\le |\C_{\u^*}|\cdot|\D| 
\le |\C_{\u^*,0}|\cdot|\C_{\u^*,1}|\cdot|\D|
\le \left(\frac{C_0}{\epsilon}+O_L(1)\right)\cdot\left(\frac{C_1}{\epsilon}+O_L(1)\right)\cdot2^{n(I(w,v,a)+\epsilon)}, \notag 
\end{align}
and
$R(\C)\le I(w,v,a) + \epsilon + o(1)$. 
Taking $ \epsilon\to0 $ finishes the proof. 
\end{proof}

Though not used in the proof of Theorem~\ref{thm:eb}, the upper bound on the size of a covering given in Lemma~\ref{lem:covering} is actually tight (on the exponential scale). 
Such a converse follows from a simple sphere-covering-type argument as below. 
\begin{lemma}
\label{lem:covering-converse}
Let $ I(w,v,a) $ and $ P_{U,X} $ be as defined in Lemma~\ref{lem:covering}. 
Then any covering $ \D\subset S_{nv}^{\mathrm{H}}(\0) $ of $ S_{nw}^{\mathrm{H}}(\0) $ satisfying the property in Lemma~\ref{lem:covering} has size at least $ 2^{nI(w,v,a) - o(n)} $. 
\end{lemma}

\begin{proof}
For any $ \u\in S_{nv}^{\mathrm{H}}(\0) $, we first compute the size of $ A(\u,P_{U,X}) $ defined in Eqn.~\eqref{eqn:A-def}. 
\begin{align}
|A(\u,P_{U,X})| &= \binom{n(1-v)}{n(w-a)}\binom{nv}{na} \le 2^{n\left[(1-v)H\left(\frac{w-a}{1-v}\right) + vH\left(\frac{a}{v}\right)\right]}. \notag 
\end{align}
Now, by the covering property of $ \D $, we have
\begin{align}
&& S_{nw}^{\mathrm{H}}(\0) &= \bigcup_{\u\in\D} A(\u,P_{U,X}) \notag \\
\implies&& |S_{nw}^{\mathrm{H}}(\0)| &\le \sum_{\u\in\D}|A(\u,P_{U,X})| \notag \\
\implies&& 2^{nH(w) - o(n)} &\le |\D| 2^{n\left[ (1-v)H\left(\frac{w-a}{1-v}\right) + vH\left(\frac{a}{v}\right)\right]} \notag \\
\implies&& |\D| &\ge 2^{n\left[ H(w) - (1-v)H\left(\frac{w-a}{1-v}\right) - vH\left(\frac{a}{v}\right)\right] - o(n)} = 2^{nI(w,v,a) - o(n)}. \qedhere 
\end{align}
\end{proof}

\subsection{Lower bound on $(L-1)$-list-decoding capacity}
\label{sec:lb-cap}


In this section, our goal is to construct an $(L-1)$-list-decodable code $ \C\subset\{0,1\}^n $ for a Z-channel with noise level $ \tau $. 
We would like to obtain a lower bound on the rate $ R(\C) $ that can be achieved by such a $\C$.

Before deriving the bound, we first introduce a set of distributions that will play an important role in the proceeding analysis. 
For any $ w\in[0,1] $ and $ \tau\in[0,w] $, define the following set of distributions:
\begin{align}
\mathcal{K}(w,\tau) &\coloneqq \left\{P_{X_1,\cdots,X_L}\in\P(\{0,1\}^L):\begin{array}{l}
\forall i\in[L],\;P_{X_i}(1) = w \\
w - P_{X_1,\cdots,X_L}(1,\cdots,1) \le \tau
\end{array}\right\}. \label{eqn:def-k} 
\end{align}
In words, $ \mathcal{K}(w,\tau) $ is the collection of distributions $ P_{X_1,\cdots,X_L} $ on length-$L$ binary strings satisfying:  
  $(i)$ each of its marginal $ P_{X_i} $ (for $ 1\le i\le L $) is $ \mathsf{Ber}(w) $;
  $(ii)$ the probability mass $ P_{X_1,\cdots,X_L}(1,\cdots,1) $ on the all-one string is at least $ w-\tau $. 

We now describe our construction and analyze it. 
Our approach follows the standard random coding with expurgation technique. 
We sample a codebook $ \C\in\{0,1\}^{M\times n} $ each entry of which is i.i.d.\ according to $ \mathsf{Ber}(w) $.

First, we claim that in expectation a $ 1/\mathrm{poly}(n) $ fraction of the codewords have weight exactly $ nw $.
To see this, note that for any $ \x\in\C $, 
\begin{align}
\Pr\{\wt(\x) = nw\} &= \binom{n}{nw}w^{nw}(1-w)^{n(1-w)} \notag \\
&= \frac{\sqrt{2\pi n}(n/e)^n}{\sqrt{2\pi nw}(nw/e)^{nw}\sqrt{2\pi(1-w)}(n(1-w)/e)^{n(1-w)}}(1-O_L(n^{-1}))w^{nw}(1-w)^{n(1-w)} \notag \\
&= \frac{1}{\sqrt{2\pi nw(1-w)}}(1-O_L(n^{-1})), \notag 
\end{align}
where the second step is by Stirling's approximation $ n! = \sqrt{2\pi n}(n/e)^n(1+O_L(n^{-1})) $. 
Therefore, 
\begin{align}
\E|\{\x\in\C:\wt(\x) = nw\}| &= \frac{M}{\sqrt{2\pi nw(1-w)}}(1-O_L(n^{-1})). \notag 
\end{align}

Second, we compute the expected number of \emph{bad} $L$-lists, i.e., those $L$-lists whose list-decoding radius is at most $ n\tau $. 
For any list $\L\in\binom{\C}{L} $ of codewords all of weight $nw$, it is clear that $ \rad(\L)\le n\tau $ (where $ \rad(\L) $ was defined in Eqn.~\eqref{eqn:def-rad-cw}) if and only if the joint type $ T_\L $ of $ \L $ is in $ \mathcal{K}(w,\tau) $. 
By Sanov's theorem (Theorem~\ref{thm:sanov}), 
\begin{align}
-\frac{1}{n}\log\Pr\{\rad(\L)\le n\tau\} 
&= -\frac{1}{n}\log\Pr\{T_\L\in\mathcal{K}(w,\tau)\} \notag \\
&\xrightarrow{n\to\infty} \min_{P_{X_1,\cdots,X_L}\in\mathcal{K}(w,\tau)} D\left(P_{X_1,\cdots,X_L}\middle\|\mathsf{Ber}^{\otimes L}(w)\right) \eqqcolon E(w,\tau). \notag 
\end{align}
Therefore, for sufficiently large $n$, we have 
\begin{align}
\E\left|\left\{\L\in\binom{\C}{L}:\rad(\L)\le n\tau\right\}\right| &= \binom{M}{L} 2^{-nE(w,\tau)(1-o(1))}. \notag 
\end{align}

Finally, we expurgate all codewords whose weights are not exactly $nw$; we also expurgate one codeword from each of the \emph{bad} $L$-lists $\L$, i.e., those $\L$ such that $ \rad(\L)\le n\tau $. 
If
\begin{align}
&& \binom{M}{L}2^{-nE(w,\tau)(1-o(1))} &\le \frac{M}{2\sqrt{2\pi nw(1-w)}}(1-O_L(n^{-1})) \notag \\
\impliedby && M^L2^{-nE(w,\tau)(1-o(1))} &\le \frac{M}{2\sqrt{2\pi nw(1-w)}}(1-O_L(n^{-1})) \notag \\
\impliedby && RL - E(w,\tau)(1-o(1)) &\le R(1-o(1)) \notag \\
\impliedby && R &\le \frac{E(w,\tau)}{L-1}(1-o(1)), \notag 
\end{align}
then after expurgation, we are left with at least $ \frac{M}{2\sqrt{2\pi nw(1-w)}}(1-O(n^{-1})) $ codewords which form an $(L-1)$-list-decodable code $ \C'\subset\C $. 
Note that the rate $R(\C')$ is asymptotically equal to $ R(\C) $.

Therefore, we have proved the following theorem. 
\begin{theorem}
\label{thm:rand-cod-lb}
There exist $w$-constant-weight $(L-1)$-list-decodable codes for the Z-channel with list-decoding radius $ \tau<\tau_L(w) $ of rate at least 
\begin{align}
\frac{1}{L-1}\min_{P_{X_1,\cdots,X_L}\in\mathcal{K}(w,\tau)} D\left(P_{X_1,\cdots,X_L}\middle\|\mathsf{Ber}^{\otimes L}(w)\right), \notag 
\end{align}
where $ \mathcal{K}(w,\tau) $ was defined in Eqn.~\ref{eqn:def-k}. 
\end{theorem}


\subsection{List-decoding capacity}
\label{sec:listdec-cap-large}
Obtaining the exact $(L-1)$-list-decoding capacity for the Z-channel is a difficult question and we are only able to derive nonmatching upper and lower bounds in Sec.~\ref{sec:ub-cap} and~\ref{sec:lb-cap} respectively. 
Nevertheless, we can compute the list-decoding capacity $ C_{L-1}(w,\tau) $ when $ L $ is sufficiently large. 
Specifically, we determine the limit $ \lim\limits_{L\to\infty}C_{L-1}(w,\tau) $ using the following two lemmas. 
The proofs below follow the outline that one uses to prove the standard list-decoding capacity theorem for symmetric errors. 

Define
\begin{align}
C_{\mathrm{LD}}(w,\tau) &\coloneqq -(1-w+\tau)\log(1-w+\tau) + \tau\log \tau - w\log w. \label{eqn:def-listdec-cap-large}
\end{align}
Note that for any fixed $w$, $ C_{\mathrm{LD}}(w,\tau) $ is convex and decreasing in $\tau$. 
It attains its maximum value $ H(w) $ at $ \tau = 0 $ and attains its minimum value $0$ at $ \tau=w $. 
Furthermore, $ C_{\mathrm{LD}}(w,\tau) $ is concave in $w$ and attains its maximum value at $ w = \frac{1+\tau}{2} $ for any fixed $\tau$. 
The corresponding maximum value is $ C_{\mathrm{LD}}(\tau)\coloneqq -(1+\tau)\log\frac{1+\tau}{2}+\tau\log\tau $ which is in turn convex and decreasing with maximum value $1$ at $ \tau = 0 $ and minimum value $0$ at $ \tau = 1 $.

\begin{lemma}[Upper bound]
\label{label:listdec-cap-large-ub}
For any $ \delta\in(0,1) $, any $w$-constant-weight code of rate $ C_{\mathrm{LD}}(w,\tau)+\delta $ for the Z-channel with error fraction $\tau$ is not $ (L-1) $-list-decodable for any $ L < 2^{n\delta+o(n)} $. 
\end{lemma}

\begin{proof}
Let $\C$ be any $w$-constant-weight code of rate $ C_{\mathrm{LD}}(w,\tau)+\delta $. 
To show non-list-decodability, we need to exhibit a \emph{bad} center $\y$ such that the Z-ball around $\y$ of radius $ n\tau $ contains at least $2^{n\delta+o(n)}$ codewords from $\C$. 
Indeed, we will show that a random $\y$ has such a property. 
Specifically, let $ \y $ be uniformly distributed among all vectors of weight $ n(w-p) $. 
We compute the expected number of codewords in $ B^{\mathrm{Z}}_{n\tau}(\y) $. 
\begin{align}
    \E|B^{\mathrm{Z}}_{n\tau}(\y)\cap\C| &= \sum_{\x\in\C} \Pr\{B^{\mathrm{Z}}_{n\tau}(\y)\ni\x\} \notag \\
    &= \sum_{\x\in\C}\Pr\{\s(\y)\subseteq\s(\x)\} \notag \\
    &= \sum_{\x\in\C}\frac{\binom{nw}{n(w-p)}}{\binom{n}{n(w-p)}} \notag \\
    &= 2^{nC_{\mathrm{LD}}(w,\tau)+n\delta}2^{n\left(wH\left(\frac{w-p}{w}\right)-H(w-p)\right)+o(n)} \notag \\
    &= 2^{n\delta+o(n)}. \notag
\end{align}
This finishes the proof. 
\end{proof}

\begin{lemma}[Construction]
\label{label:listdec-cap-large-lb}
For any $ \delta\in(0,1) $ and any $ L>1/\delta $, there exist $ (L-1) $-list-decodable codes for the Z-channel with error fraction $\tau$ of rate $ C_{\mathrm{LD}}(w,\tau)-\delta $. 
\end{lemma}

\begin{proof}
First note that if $ \tau\ge w $, the capacity is trivially zero under any constant list size since the channel can zero out all bits in any codeword and the list size has to be as large as the code size. 
In the proof we therefore assume $ \tau<w $. 

Let $ \C\subset\{0,1\}^n $ be a collection of $ M\coloneqq2^{nR} $ (where $R = C_{\mathrm{LD}}(w,\tau) - \delta$) random vectors $ \x_1,\cdots,\x_M $ each of which is i.i.d.\ chosen from the set of all weight-$ nw $ binary vectors. 
If a codeword $\x\in\C$ is input to the channel, the weight of the output vector $ \y\in\{0,1\}^n $ is in the range $ [n(w-\tau),nw] $. 
A vector $ \y\in\{0,1\}^n $ of weight in $ [n(w-\tau),nw] $ can result from $ \x\in\C $ through the channel if and only if $ \s(\y)\subset\s(\x) $. 
We compute the probability that this happens to a random codeword $ \x\in\C $. 
For any fixed $ \y\in\{0,1\}^n $ of weight $ nu $ where $ u\in[w-\tau,w] $, we have 
\begin{align}
\Pr\{\s(\y)\subseteq\s(\x)\} &= \frac{\binom{n(1-u)}{n(w-u)}}{\binom{n}{nw}} 
\le 2^{n(1-u)H\left(\frac{w-u}{1-u}\right) - nH(w) + o(n)}. \notag 
\end{align}
Therefore, the probability that such a random code $ \C $ is not $(L-1)$-list-decodable is:
\begin{align}
& \Pr\{\C\text{ is not $(L-1)$-list-decodable}\} \notag \\
&\le \Pr\left\{\exists \{i_1,\cdots,i_L\}\in\binom{[M]}{L},\;\exists\y\in\{0,1\}^n,\;\mathrm{s.t.}\;\wt(\y)\in[n(w-\tau),nw];\;\forall j\in[L],\;\s(\y)\subseteq\s(\x_{i_j})\right\} \notag \\
&\le \binom{M}{L}2^n\max_{w-\tau\le u\le w} \left(2^{n(1-u)H\left(\frac{w-u}{1-u}\right) - nH(w) + o(n)}\right)^L \notag \\
&\le \max_{w-\tau\le u\le w} 2^{nRL + n + nL\left[(1-u)H\left(\frac{w-u}{1-u}\right) - H(w) + o(1)\right]} \notag \\
&\le 2^{n\left[ RL+1+L\left((1-w+\tau)H\left(\frac{\tau}{1-w+\tau}\right)-H(w) +o(1)\right)\right]}. \label{eqn:listdec-cap-exponent} 
\end{align}
The last inequality follows since the function $ f_w(u) \coloneqq (1-u)H\left(\frac{w-u}{1-u}\right) $ has the following property. 
For any fixed $ w $, $ f_w(u) $ is concave and decreasing in $ u\in[0,w] $, and therefore in the domain $ u\in[w-\tau,w] $ it attains its maximum at $ u = w-\tau $. 

By elementary algebraic manipulation, we have 
\begin{align}
(1-w+\tau)H\left(\frac{\tau}{1-w+\tau}\right)-H(w) = -C_{\mathrm{LD}}(w,\tau). \notag 
\end{align}
Recall $ R = C_{\mathrm{LD}}(w,\tau) - \delta $. 
Then the exponent (normalized by $ n^{-1} $) of the RHS of the inequality~\eqref{eqn:listdec-cap-exponent} equals 
\begin{align}
(C_{\mathrm{LD}}(w,\tau) - \delta)L+1-C_{\mathrm{LD}}(w,\tau)L+o(1) = 
-\delta L + 1 + o(1), \notag 
\end{align}
which is negative if $ L>1/\delta + o(1) $. 
That is, the probability that $\C$ is $(L-1)$-list-decodable is at least $ 1-2^{-\Omega(n)} $ for any $ L>1/\delta + o(1) $. 
\end{proof}

Lemma~\ref{label:listdec-cap-large-ub} and~\ref{label:listdec-cap-large-lb} imply the following characterization. 
\begin{theorem}[List-decoding capacity]
\label{thm:listdec-cap-large}
The $(L-1)$-list-decoding capacity of the Z-channel with error fraction $\tau$ under input constraint $ w $ is given by $ \lim\limits_{L\to\infty}C_{L-1}(w,\tau) = C_{\mathrm{LD}}(w,\tau) $ as defined in Eqn.~\eqref{eqn:def-listdec-cap-large}. 
\end{theorem}

\section{Capacity of stochastic Z-channels}
\label{sec:cap-stoch-z}
In all other sections of this paper, we considered Z-channels with \emph{adversarial} errors. 
In this section, we derive the capacity of \emph{stochastic} Z-channels from Shannon's seminal channel coding theorem. 

To this end, we first state Shannon's theorem for general discrete memoryless channels (DMCs). 
Let $ \X $ and $ \Y $ be two finite sets denoting the input and output alphabets of the channel. 
A DMC is a stochastic matrix $ W_{Y|X} $ that maps a distribution over $ \X $ to a distribution over $ \Y $. 
When the channel is used $n$ times with an input sequence $ \x\in\X^n $, the output sequence $ \y\in\Y^n $ follows a product law: 
\begin{align}
\Pr\{\y = (b_1,\cdots,b_n)|\x = (a_1,\cdots,a_n)\} &= \prod_{i = 1}^n W_{Y|X}(b_i|a_i), \notag 
\end{align}
for any $ (a_1,\cdots,a_n)\in\X^n $ and $ (b_1,\cdots,b_n)\in\Y^n $. 

{
For a code $ \C = \{\x_1,\cdots,\x_M\}\subset\X^n $, its \emph{rate} is defined as $ R(\C) \coloneqq\frac{1}{n}\log M $. 
The \emph{average (over messages) probability of error} of $\C$ when used over a DMC $ W_{Y|X} $ with a decoder $ \mathrm{Dec}\colon\Y^n\to[M] $ is defined as 
\begin{align}
P_{\mathrm{e},\mathrm{avg}}(\C) &\coloneqq \frac{1}{M} \sum_{i = 1}^M \sum_{y\in\Y^n} \prod_{j = 1}^n W_{Y|X}(y_j|x_{i,j}) \mathds1\{\mathrm{Dec}(y)\ne i\} . \notag 
\end{align}
The \emph{capacity} of the channel is then defined as 
\begin{align}
C(W_{Y|X}) &\coloneqq \limsup_{\epsilon\downarrow0} \limsup_{n\uparrow\infty} \max_{\substack{\C_n\subset\X^n \\ P_{\mathrm{e},\mathrm{avg}}(\C_n)\le\epsilon}} R(\C_n) , \notag 
\end{align}
i.e., the largest rate for which there exists a sequence of codes with vanishing error probability. 
}
The DMC is said to be equipped with input constraints $ \Q\subset\P(\X) $ if the type $ T_{\x} $ of any input sequence $\x\in\C $ is required to belong to $ \Q $. 
{The capacity of $ W_{Y|X} $ with input constraints $\Q$ can be similarly defined. }

\begin{theorem}[Channel coding theorem~\cite{shannon1948mathematical}]
The capacity $C(W_{Y|X})$ of a DMC $ W_{Y|X}\in\P(\Y|\X) $ with input constraints $ \Q $ is given by 
\begin{align}
C(W_{Y|X}) &= \max_{P_X\in\Q} I(X;Y), \label{eqn:cap-dmc} 
\end{align}
where the mutual information is evaluated {with respect to} the joint law $ P_XW_{Y|X} $. 
\end{theorem}

A discrete memoryless Z-channel is defined as follows. 
Both the input and output alphabets are binary: $ \X = \Y = \{0,1\} $. 
The channel transition law is parameterized by the zeroing-out probability $ \tau $, i.e., 
$ W^{\mathrm{Z}}_{Y|X}(0|0) = 1, W^{\mathrm{Z}}_{Y|X}(0|1) = \tau $. 
The input constraint is such that all input sequences should have Hamming weight at most $ nw $ for some $ w\in[0,1] $. 
It is easy to evaluate Eqn.~\eqref{eqn:cap-dmc} which yields
$C(W^{\mathrm{Z}}_{Y|X}) = C(w,\tau) = H(w(1-\tau)) - wH(\tau)$. 

We note that for any fixed $ w\in[0,1] $, $ C(w,\tau) $ is convex and decreasing in $\tau$. 
It attains its maximum value $ H(w) $ at $\tau = 0$ and attains its minimum value $ 0 $ at $ \tau = 1 $. 

We also note that $ C(w,\tau) $ is concave in $w$ for any fixed $\tau$. 
The maximizing $w$ is given by 
\begin{align}
w_{\max} &\coloneqq \left(1 - \tau + \tau^{-\frac{\tau}{1-\tau}}\right)^{-1}, \notag
\end{align}
{and the corresponding capacity is given by $ C(\tau) \coloneqq H(w_{\max}(1-\tau)) - w_{\max}H(\tau) $. 
Equivalent expressions have also been presented in \cite{tallini2008feedback}: }
\begin{align}
w_{\max} &= \left( \left(1 + 2^{\frac{H(\tau)}{1-\tau}}\right) (1-\tau) \right)^{-1} , \quad 
C(\tau) = \log\left(1 + \tau^{\frac{\tau}{1-\tau}} - \tau^{\frac{1}{1-\tau}}\right). \notag 
\end{align}
Note that $ C(\tau) $ is convex and decreasing with maximum value $1$ at $ \tau = 0 $ and minimum value $ 0 $ at $ \tau = 1 $.

For symmetric errors and erasures, it happens that the {capacities} of binary symmetric channels and binary erasure channels coincide with the respective list-decoding {capacities} (under adversarial errors). 
Comparing the expressions of $C(w,\tau)$ and $C_{\mathrm{LD}}(w,\tau)$, we see that this is no longer true for Z-channels. 

\section{Open problems}
\label{sec:open}
\begin{enumerate}
    \item For codes correcting symmetric errors, the Elias--Bassalygo bound can be improved using Delsarte's linear program \cite{mceliece1977new}. 
    We do not know how to derive a linear programming-type bound for Z-channels. 
    As far as we know, the linear programming framework, in its most general form, assumes that the ambient space that the code lives in (which is $ \{0,1\}^n $ for general codes and $ S_{nw}^{\mathrm{H}}(\0) $ for $w$-constant-weight codes) can be defined as an \emph{association scheme}. 
    We do not see how to do so under the Z-metric $ d_{\mathrm{Z}}(\cdot,\cdot) $ since the volume of the intersection of two Z-spheres is not invariant under shifts. 
    
    \item The largest code size is exponential in $n$ if the fraction of asymmetric errors the list-decodable code can correct is less than the Plotkin point $ \tau_L $ and we gave bounds on the exponent in Sec.~\ref{sec:ub-cap} and~\ref{sec:lb-cap}; whereas it is $ \Theta_L(\epsilon^{-3/2}) $ if the fraction of errors is $\epsilon$-above $\tau_L$. 
    There is one missing case which we did not solve, that is, what is the largest code size with error fraction being \emph{exactly} $ \tau_L $?
    We conjecture that in this case the answer is $ \Theta_L(n^{3/2}) $. 
    Note that the answer to the same question for symmetric errors is $ 2n $ proved by a geometric argument. 
\end{enumerate}

\bibliographystyle{alpha}
\bibliography{ref}
\end{document}